\documentclass[manuscript,nonacm]{acmart}
\usepackage{booktabs}  % for top/mid/bottom rules
\usepackage{makecell}  % for line breaks in table headers
\usepackage{amsmath}   % for math like \mathcal{O}
\usepackage{algorithm}%
\usepackage{algorithmicx}%
\usepackage{algpseudocode}%

\newtheorem{remark}{Remark}
\newcommand{\simdot}{\mathrel{\sim^{\cdot}}}
\newcommand{\candP}[1]{\mathsf{P}_{#1}}   
\newcommand{\pLTS}[1]{\mathcal{P}_{#1}}

\AtBeginDocument{%
  }

\begin{document}

\title{Semantic Fusion: Verifiable Alignment in Decentralized Multi-Agent Systems}

\author{Sofiya Zaichyk}
\orcid{}
\affiliation{%
  \institution{Innovative Defense Technologies}
  \city{Arlington}
  \state{VA}
  \country{USA}
}

\begin{abstract}
We present \emph{Semantic Fusion (SF)}, a formal framework for decentralized semantic coordination in multi-agent systems. SF allows agents to operate over scoped views of shared memory, propose structured updates, and maintain global coherence through local ontology-based validation and refresh without centralized control or explicit message passing. The central theoretical result is a bisimulation theorem showing that each agent’s local execution is behaviorally equivalent to its projection of the global semantics, in both deterministic and probabilistic settings. This enables safety, liveness, and temporal properties to be verified locally and soundly lifted to the full system. SF supports agents whose update proposals vary across invocations, including those generated by learned or heuristic components, provided updates pass semantic validation before integration. We establish deterministic and probabilistic guarantees ensuring semantic alignment under asynchronous or degraded communication. To validate the model operationally, we implement a lightweight reference
architecture that instantiates its core mechanisms. A 250-agent simulation evaluates these properties across over 11,000 validated
updates, demonstrating convergence under probabilistic refresh, bounded communication, and resilience to agent failure. Together,
these results show that Semantic Fusion provides a formal and scalable basis for verifiable autonomy in decentralized systems.
\end{abstract}




\maketitle

\section{Introduction}\label{sec:intro}

Rising system complexity and increasing coordination demands have exposed a longstanding tension between adaptability and verifiability in autonomous and multi-agent architectures. Symbolic systems provide structure, interpretability, and formal guarantees, but often fail to operate reliably under dynamic or asynchronous conditions. Neural components, including learned policies or language-model–based interpreters, offer flexibility and generalization but tend to produce outputs that lack the explicit structure needed for traceability, compliance, and modular verification. These limitations are especially acute in decentralized settings, where agents must reason and act independently without centralized control or global observability.

In many practical multi-agent systems, such as distributed sensing and robotic warehouse coordination, agents operate under limited bandwidth, partial observability, and asynchronous decision cycles. Increasingly, these agents rely on learned modules to process structured and unstructured input. Although such modules can improve adaptability, their outputs may be ambiguous or inconsistent in terminology and often lack clear provenance. In settings without centralized control or globally shared state, agents must coordinate using only partial information while ensuring that decisions remain aligned and verifiable. This motivates the need for a principled framework that supports decentralized semantic consistency, scoped reasoning, and provable alignment, even when agent behavior is adaptive or involves nondeterministic update generation.

This paper introduces \emph{Semantic Fusion}, a formal framework for decentralized semantic alignment in multi-agent systems. Each agent operates over a constrained projection of shared memory (its \emph{slice}), maintains a partial view of evolving system state, and proposes structured updates that are subject to local semantic validation. Coherence and alignment follow from how agents interact with shared state, rather than being externally imposed. The architecture avoids explicit message passing and global synchronization, enabling coordination through disciplined memory interaction alone.

Robustness under degraded or asynchronous conditions forms a core motivation for this work. In many operational environments, communication loss, adversarial interference, and uncertainty are expected rather than exceptional. Coordination models that rely on centralized control or globally visible state often fail in such settings. By enforcing semantic coherence through scoped memory interaction, Semantic Fusion enables decentralized operation even when connectivity fluctuates or observability becomes unreliable.

Our central theoretical result is a bisimulation theorem that establishes behavioral equivalence between each agent’s local execution and the global system semantics projected onto its slice. This enables safety and temporal properties to be verified locally and soundly extended to the full system, addressing a common difficulty in connecting localized reasoning to global guarantees. The framework yields deterministic guarantees such as semantic coherence, convergence of agent-local memories, and semantically scoped causal isolation. We further extend the model to accommodate nondeterministic update generators, showing that global safety properties remain compositionally verifiable from agent-level traces. A final result establishes almost-sure liveness, ensuring that agents’ local memories converge to the evolving global semantic state over time.

This makes it possible to design agent architectures where global constraints, including safety rules, temporal dependencies, or policy boundaries, can be enforced through local validation and scoped memory access, without requiring global state awareness. While principles such as scoped memory access or bisimulation appear individually in prior work, their combination here enables a framework that supports deterministic and probabilistic reasoning with provable global alignment through slice-scoped semantic interaction. This includes agents whose decision processes incorporate learned modules, provided their outputs undergo local semantic validation before affecting shared state.

To demonstrate feasibility, we implement a reference architecture that instantiates the model’s core mechanisms. While not a contribution in itself, the system provides an executable validation of the theoretical framework under dynamic, asynchronous conditions. A 250-agent simulation in a search-and-rescue scenario illustrates the behaviors predicted by the formal results and shows that coordinated global behavior can arise from interaction with shared semantic state rather than centralized control.

Together, these results show that a coherent and eventually convergent global execution can be realized through slice-scoped validation and structured state updates, without centralized memory, global clocks, or direct messaging, and with formal alignment to projected global semantics. Although we refer to a global semantic memory for analysis, the system does not require a centralized store. Its semantics arise from the distributed ensemble of local executions, which collectively realize the global trace.

The remainder of the paper is organized as follows. Section~\ref{sec:related-work} surveys related work. Section~\ref{sec:core-constructs} introduces the core formal constructs. Section~\ref{sec:semantic-fusion} presents the reference architecture. Section~\ref{sec:formal-model} develops the formal properties and proofs. Section~\ref{sec:experimental-val} provides empirical validation. Section~\ref{sec:discussion} discusses deployment, extensions, and future directions. Section~\ref{sec:conclusion} concludes.

\section{Related Work}\label{sec:related-work}
Early efforts in multi-agent coordination were shaped by symbolic architectures grounded in explicit representations of meaning. Blackboard systems such as the Hearsay-II speech recognizer~\cite{erman1980hearsay} and cognitive frameworks like SOAR~\cite{laird1987soar} offered interpretable control through rule-based activation over shared memory spaces. However, these systems often proved brittle in dynamic or adversarial environments, where fixed rules limited adaptability and scalability~\cite{russell2009}.

Connectionist approaches—particularly deep learning and reinforcement learning agents~\cite{lecun2015deep}—brought flexibility through data-driven inference, but often lacked semantic grounding. Foundation models~\cite{brown2020gpt3, chowdhery2023palm, touvron2023llama} extended these capabilities to language, planning, and tool use, but remain prone to hallucination and semantic inconsistency~\cite{bender2021dangers, bommasani2021opportunities}, especially in settings requiring inter-agent coordination. In distributed multi-agent systems (MAS), this lack of persistent shared semantics introduces ambiguity and inconsistency that complicates the preservation of global invariants.

Recent orchestration frameworks such as LangChain, LangGraph, and HuggingFace Agents~\cite{langchain, langgraph, huggingfaceagents} attempt to bridge symbolic and connectionist paradigms by composing LLM workflows into tool-chaining pipelines. These platforms support dynamic behavior through prompt templates and modular controllers, but remain episodic in nature, rely heavily on sensitive prompt engineering, and lack persistent, verifiable state. Critically, they provide no formal grounding in semantic coherence or inter-agent alignment.

Closer to our setting, SHIMI~\cite{helmi2025shimi} and DAMCS~\cite{yang2025llm} introduce decentralized memory structures for LLM-based agents. SHIMI uses a hierarchical semantic store; DAMCS proposes a distributed graph memory. Both advance semantic coordination, but neither supports structured update validation, scoped propagation, formal correctness guarantees, or liveness under probabilistic refresh. Our model differs by enforcing ontology-conformant validation and scoped propagation, allowing LLM-based agents to align and coordinate under asynchronous, partial observability with formal guarantees.

Conflict-free replicated data types (CRDTs)~\cite{shapiro2011crdt} offer a separate foundation for distributed memory, ensuring convergence under concurrent writes. However, CRDTs assume reliable delivery of all operations, disseminate updates indiscriminately, and offer no semantic scoping or modular guarantees. In contrast, our framework enforces local validation and slice-scoped update propagation, allowing agents to converge probabilistically to semantically relevant state without requiring delivery of all intermediate updates. This supports modular reasoning through agent-level bisimulation and per-element convergence theorems under partial delivery.

Our work also diverges from classical knowledge-based programs (KBPs)~\cite{fagin2003knowledge}, which use modal logic and epistemic conditions to govern agent behavior under fixed communication protocols. While KBPs support sound knowledge-driven decision making, they lack structured semantic memory and do not support modular reasoning over asynchronous state evolution. We instead formalize agent interaction through a shared, ontology-scoped memory system where correctness is ensured by local validation rather than message-passing logic.

Table~\ref{tab:comparison} summarizes these distinctions. While prior frameworks emphasize replicated memory or informal composition mechanisms, our work introduces the first formal foundation for verifiable decentralized coordination over structured semantic state. It enables constrained update propagation and local verification of global safety properties, while maintaining coherent semantic evolution under probabilistic proposals.

\setlength{\tabcolsep}{4pt}
\begin{table}[ht]
\centering
\caption{Comparison with Prior Frameworks}
\begin{tabular}{@{}lccccc@{}}
\toprule
\textbf{Feature} & \textbf{CRDTs} & \textbf{SHIMI} & \textbf{DAMCS} & \textbf{KBPs} & \makecell{\textbf{Semantic}\\\textbf{Fusion}} \\
\midrule
\multicolumn{6}{l}{\emph{Memory and Semantics}} \\
\quad Ontology-scoped memory & No & Yes & Yes & Implicit & Yes \\
\quad Structured validation & No & No & No & Implicit & Yes \\
\quad Slice-local reasoning & No & Yes & Yes & No & Yes \\
\quad Causal isolation & No & No & No & No & Yes (Thm.~\ref{thm:causal-isolation}) \\
\quad Failure resilience & Partial & No & No & No & Yes (Prop.~\ref{prop:local-failure}) \\
\midrule
\multicolumn{6}{l}{\emph{Formal Guarantees}} \\
\quad Semantic coherence & Partial & No & No & Yes (under model) & Yes (Thm.~\ref{thm:semantic-coherence}) \\
\quad Temporal bisimulation & No & No & No & No & Yes (Thm.~\ref{thm:slice_bisim}) \\
\quad Probabilistic coherence & No & No & No & No & Yes (Thm.~\ref{thm:prob_slice_bisim}) \\
\quad Global property completeness & No & No & No & No & Yes (Thm.~\ref{thm:slice-completeness-safety}) \\
\quad Almost-sure slice convergence & No & No & No & No & Yes (Thm.~\ref{thm:dynamic-convergence}) \\
\midrule
\multicolumn{6}{l}{\emph{Scalability}} \\
\quad Minimal communication & No & No & No & No & Yes (Thm.~\ref{thm:comm_lower_bound}) \\
\quad Message complexity & $\mathcal{O}(n)$ & $\mathcal{O}(n)$ & $\mathcal{O}(n)$ & $\mathcal{O}(n)$ & $\Theta(d)$ \\
\bottomrule
\end{tabular}
\label{tab:comparison}
\end{table}

\section{Core Constructs}
\label{sec:core-constructs}

To formalize slice-scoped interaction in multi-agent systems, we begin by introducing the constructs that determine how agents access shared semantic memory and validate the updates they propose. Together, these define the structure of memory access and observation, and they govern how updates are checked and propagated in a decentralized setting. They form the foundation for the formal results in Section 5 and are later instantiated in a concrete system for empirical validation (Section~\ref{sec:semantic-fusion}).

\subsection{Ontologies and Semantic Memory}
\label{subsec:ontologies}

At the heart of this framework is a global ontology $\mathcal{O}$ that
defines the system’s conceptual universe: types, relationships, constraints,
and inheritance structure. It acts as both a schema and a constraint system
for all memory evolution.

The shared semantic memory $\mathcal{M}(t)$ at time $t$ is a set of grounded,
ontology-compliant statements reflecting the current state of the world as
known to the system. Unlike message logs or static knowledge graphs, this memory is evolving and validated, and agents interact with it through slice-scoped projections rather than as a passive store.

\subsection{Ontology Slices and Semantic Projections}
\label{subsec:ontology-slices}

To support modularity and decentralization, every agent $a \in A$ is granted
an ontology slice $O_a \subseteq O$ that bounds the semantic entities it may
read from or write to in the shared memory. The agent may reason over
additional private state, but any update that crosses the shared boundary
must validate against $O_a$.

\begin{definition}[Memory Projection]
Given a global memory $\mathcal{M}(t)$ and an ontology slice $O_a$, the memory
projection for agent $a$ is
\[
\pi_{O_a}(\mathcal{M}(t)) = \{ s \in \mathcal{M}(t) \mid \mathrm{dom}(s) \in
O_a \},
\]
where $\mathrm{dom}(s)$ is the root entity of the statement. This projection
isolates the portion of the shared state relevant to agent $a$.
\end{definition}

\begin{definition}[Semantic Slice]
At time $t$, the agent’s local view is the semantic slice
\[
S_a(t) = (O_a, M_a(t)), \qquad M_a(t) \subseteq \pi_{O_a}(\mathcal{M}(t'))
\]
for some $t' \le t$. The local memory $M_a(t)$ may lag the global projection
due to asynchronous refresh propagation. This decoupling is what enables
coordination without strict synchronization.
\end{definition}

\subsection{Update Proposals and Validation}
\label{subsec:update-proposals}

Agents coordinate by reasoning over their current slice and proposing ontology-governed updates that reflect intended or observed changes.
Physical actions may occur independently, but any effects an agent chooses to represent in shared memory must be introduced through subsequent validated update proposals.

\begin{definition}[Update Proposal]
A structured proposal $P_a$ is a partial function
\[
P_a : S_a(t) \rightarrow \Delta S_a
\]
mapping the agent’s current slice to a proposed set of changes $\Delta S_a$.
\end{definition}

Not all proposals are safe or meaningful. To maintain semantic coherence, all
updates must be validated against the shared ontology.

\begin{definition}[Validation]
A proposal $P_a$ is valid if
\[
\mathcal{O} \vdash P_a .
\]
Only validated updates $\Delta S_a^{\mathrm{valid}}$ may be integrated into
$\mathcal{M}(t)$.
\end{definition}

This guarantee holds regardless of whether an agent’s internal decision mechanism is symbolic, learned, or generative, because only ontology-conformant, locally validated updates enter the global state. The ontology acts as a
semantic gatekeeper for decentralized coordination.

\subsection{Refresh Propagation and Scoped Synchronization}
\label{subsec:refresh-propagation}

Once an update has been validated and incorporated into shared memory, the
framework ensures that refresh notifications are scoped to agents whose
slices intersect the modified entities, while agents with no overlap remain
unaffected.

\begin{definition}[Refresh Notification]
\label{def:refresh}
Let $\Delta S_a$ be a validated update from agent $a$. The refresh
notification is
\[
\tau(\Delta S_a) \rightarrow \mathcal{E},
\]
where $\mathcal{E}$ is the set of modified entities. Agents with $O_b \cap
\mathcal{E} \ne \varnothing$ may retrieve and locally revalidate $\Delta
S_a$.
\end{definition}

\begin{remark}
In the deterministic semantics, all agents implement the same
ontology-scoped validation predicate $V$ restricted to their slices.
Consequently, if a validated update $\Delta S_a$ is retrieved by any agent
whose slice intersects its scope, slice-local revalidation succeeds; the
failure case is excluded by assumption. Local revalidation remains part of
the semantics because safety is enforced slice-locally: agents do not have
global state and must independently confirm that an update is admissible
for their slice.

In the probabilistic extension and in practical deployments, validators
may be imperfect, stale, or approximate, and an agent may occasionally
reject a globally valid update. In such cases the agent performs a semantic
stutter step, leaving its local memory unchanged. Theorems~5.24 and~5.25
establish that such disagreements are transient with high probability.
\end{remark}

This selective refresh mechanism ensures that agents synchronize only on
changes relevant to their semantic domain. It enables causal isolation,
bounded communication, and eventual slice convergence without central
coordination.

\subsection{Constructive Summary}
\label{subsec:constructive-summary}

These constructs define more than a vocabulary; they specify a coherent model for decentralized semantic reasoning. Ontology slices shape how each agent contributes to shared state by constraining how perceptions are represented and integrated. Semantic slices track local memory as it evolves through selectively accepted updates. Structured proposals capture intended changes, and validation enforces their coherence with system-wide rules. Refresh propagation maintains relevance without requiring global synchronization. The framework imposes no requirements on an agent’s internal policy; it constrains only the structure and admissibility of the updates the agent contributes.

This model reframes coordination by grounding interaction in ontology-scoped memory rather than direct belief exchange. By constraining how meaning evolves, it enables formal reasoning, modular verification, and traceable coordination. Although ontology-governed memory may restrict expressivity, it provides the structure needed for rigorous correctness claims in decentralized settings. Future work may examine hybrid designs that incorporate both structured and unstructured state, permitting selective relaxation while preserving global guarantees.

The full formal model comprising these constructs is referred to as Semantic
Fusion (SF) throughout this paper. In Section~\ref{sec:semantic-fusion} we
instantiate these constructs in a lightweight reference system, not as the
core contribution but as a confirming mechanism that demonstrates the model’s
operational feasibility under realistic conditions.

\section{Reference Architecture: An Instantiation of Semantic Fusion (SF)}
\label{sec:semantic-fusion}

To assess the operational feasibility of the Semantic Fusion (SF) model, we implemented a lightweight reference system that instantiates SF’s core constructs: scoped slices, local validation, and selective refresh. The architecture provides an executable realization of the model’s semantics and enables empirical confirmation of its behavior under asynchronous, dynamic conditions. The architecture describes a shared-memory abstraction for convenience, but the semantics do not depend on a centralized memory structure. Any architecture that preserves slice-scoped validation and refresh induces the global semantics.

\subsection{Agent Architecture}

Each agent comprises three modular subsystems aligned with the formal model. The \textit{slice manager} maintains the agent’s scoped semantic slice 
$S_a(t) = (O_a, M_a(t))$, processes refresh notifications, and determines which entities fall within the agent’s ontology-restricted view. The \textit{inference engine} performs local reasoning over $M_a(t)$ using arbitrary internal mechanisms (e.g.\ symbolic rules, probabilistic policies, or hybrid decision logic). The \textit{update generator} constructs structured update proposals $P_a$ and applies ontology-scoped validation; if $O \vdash P_a$, the update is integrated and triggers a refresh notification to relevant peers.

Agents operate fully asynchronously and without centralized coordination. No global clock or ordering primitive is assumed; semantic alignment emerges through slice-scoped interaction with shared memory. A schematic overview is provided in Figure~\ref{fig:arch}.

\begin{figure}[ht]
  \centering
  \includegraphics[width=0.85\linewidth]{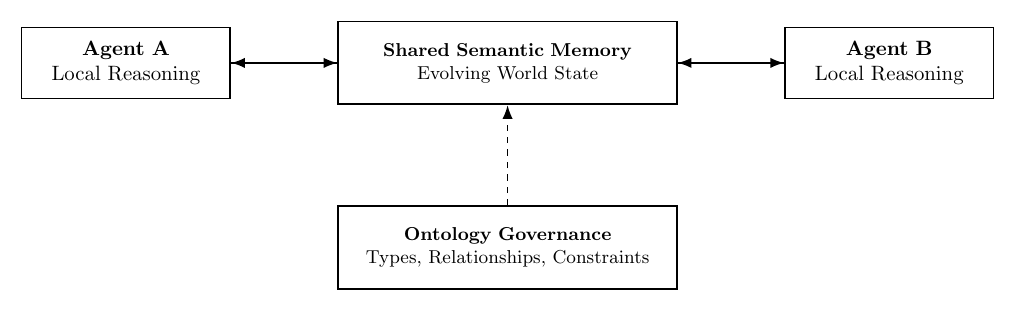}
  \caption{Reference architecture instantiating Semantic Fusion (SF). 
  Agents operate over scoped semantic slices, proposing and validating structured updates to shared semantic memory. 
  Coherence and coordination arise from ontology-scoped validation and refresh propagation rather than explicit messaging or global control.}
  \label{fig:arch}
\end{figure}

\subsection{Lifecycle and Communication}

Each agent follows an asynchronous execution loop. Upon receiving a refresh notification, the slice manager determines whether the referenced entities intersect its slice; if so, the agent selectively retrieves the corresponding update, re-validates it, and integrates it into local memory. All reasoning occurs over the slice $S_a(t)$, enabling the agent to infer goals or generate proposals based solely on ontology-scoped information.

When an agent generates a valid update proposal, it commits the change locally and issues a scoped notification containing only the identifiers of the modified entities. Agents whose slices intersect those identifiers pull and validate the update independently. This selective mechanism minimizes communication, enforces causal isolation (Theorem~\ref{thm:causal-isolation}), and preserves the trace semantics needed for local–global bisimulation.

The reference prototype implements a minimal scoped-delivery mechanism sufficient to realize the semantics of update propagation and retrieval described here. More sophisticated dissemination strategies, such as content-based publish/subscribe systems, could be used in practical deployments, but lie outside the scope of this work. See Appendix~\ref{app:ref-arch} for a high-level outline of the execution algorithm.

\subsection{Hybrid Reasoning and Validation}

The reference architecture is agnostic to the agent’s internal reasoning strategy. 
Agents may employ symbolic logic, probabilistic decision modules, or learning-based components. 
All proposals, however, pass through the same ontology-scoped validation interface before integration. 
This decoupling enables heterogeneous agents to operate over a shared semantic substrate while maintaining safety. To accommodate stochastic or uncertain proposal mechanisms, agents may employ probabilistic validators as formalized in Theorem~\ref{thm:prob-coherence}, allowing uncertain proposals to be handled in a principled way while preserving semantic guarantees.

\subsection{Architecture Summary}

The SF reference architecture provides a concrete instantiation of the theoretical model presented in Section~\ref{sec:core-constructs}. 
It supports asynchronous execution and enforces ontology-scoped validation with selective refresh. The architecture is not intended to optimize performance; its role is to show that constrained semantic reasoning can be realized in a decentralized setting and that the resulting system behavior matches the formal analysis. The guarantees developed in the next sections are then exercised empirically in a dynamic multi-agent simulation.

\section{Formal Model and Properties}\label{sec:formal-model}

This section develops the formal model of Semantic Fusion. We show that decentralized agents operating over constrained projections of shared state can maintain global coherence and behavioral alignment through structured updates and local validation, without centralized coordination. The results indicate that ontology-scoped memory interaction is sufficient to sustain the global invariants required for correctness.

We organize the results into three families:
\begin{itemize}
    \item \textbf{Deterministic guarantees:} semantic coherence, slice convergence, causal isolation, and communication bounds;
    \item \textbf{Behavioral equivalence:} a slice-global stuttering bisimulation that supports modular temporal reasoning;
    \item \textbf{Probabilistic guarantees:} probabilistic bisimulation of agent executions and resilience to imperfect update validation.
\end{itemize}
 
Theorems~\ref{thm:semantic-coherence}-
 \ref{thm:prob_slice_bisim} assume reliable scoped delivery (Assumption S2, Appendix~\ref{app:proof-thm10}), as would be enforced by a publish-subscribe infrastructure or reliable multicast. Theorem~\ref{thm:dynamic-convergence} relaxes this to probabilistic delivery with i.i.d. drop behavior and proves eventual convergence.

Each result is stated and sketched in this section. Full derivations and proofs of key theorems are provided in Appendix~\ref{app:proofs}. We distinguish between foundational system guarantees (labeled as Theorems), derived consequences (Corollaries), helper results (Lemmas), and structural or performance-related invariants (Propositions).

\subsection{Preliminaries and Notation}\label{subsec:preliminaries}

We consider a system of autonomous agents \( \mathcal{A} = \{ a_1, a_2, \ldots, a_n \} \) operating over a shared semantic memory \( \mathcal{M}(t) \), structured according to a global ontology \( \mathcal{O} \). Each agent \( a \in \mathcal{A} \) maintains a local memory \( M_a(t) \), scoped by an ontology slice \( O_a \subseteq \mathcal{O} \). Ontological semantics, slice construction, projections, update proposals, and refresh protocols are defined in Section~\ref{sec:core-constructs}. The symbols used in this section are summarized in Table~\ref{tab:notation}.

\begin{table}[ht]
\centering
\caption{Core Notation for Semantic Fusion Formal Model}
\begin{tabular}{@{}lp{0.7\linewidth}@{}}
\toprule
\multicolumn{2}{l}{\textbf{Agents and Ontology}} \\
\midrule
$A = \{a_1, a_2, \dots, a_n\}$ & Set of autonomous agents \\
$\mathcal{U}$ & Universe of semantic entities \\
$\mathcal{O}$ & Global ontology over $\mathcal{U}$ \\
$O_a \subseteq \mathcal{O}$ & Ontology slice of agent $a$ \\

\multicolumn{2}{l}{\textbf{Memory and State}} \\
\midrule
$\mathcal{M}(t)$ & Global semantic memory at time $t$ \\
$M_a(t)$ & Local memory of agent $a$ \\
$\pi_{O_a}(\mathcal{M})$ & Projection of $\mathcal{M}$ onto $O_a$ \\
$S_a(t) = (O_a, M_a(t))$ & Semantic slice of agent $a$ \\

\multicolumn{2}{l}{\textbf{Updates and Execution}} \\
\midrule
$P_a$ & Structured update proposals by agent $a$ \\
$\Delta S_a$ & Proposed modifications \\
$\Delta S_a^{\text{valid}}$ & Ontology-validated update \\
$\tau(\Delta S_a)$ & Refresh notification for affected entities \\
$\gamma_a$ & Probabilistic update generator for agent $a$ \\

\multicolumn{2}{l}{\textbf{Labeled Transition Systems}} \\
\midrule
$\mathcal{G}$ & Global LTS over $\mathcal{M}(t)$ \\
$\mathcal{T}_a$ & Agent $a$'s local LTS over $M_a(t)$ \\
$\mathcal{G}|_{O_a}$ & Projection of $\mathcal{G}$ onto $O_a$ \\
\bottomrule
\end{tabular}
\label{tab:notation}
\end{table}

\subsection{Deterministic Guarantees}\label{subsec:deterministic-guarantees}
We begin by establishing that the system preserves semantic coherence over time, provided that all updates are validated against the shared ontology before integration.

\begin{lemma}[Eventual Delivery of Scoped Updates]\label{lem:delivery}
Let $\Delta S_b \subseteq \mathcal{M}(t)$ be a validated update, and let $a$ be an agent with ontology slice $O_a \subseteq \mathcal{O}$. If $\text{dom}(\Delta S_b) \cap O_a \ne \varnothing$, then:
\begin{equation}
\quad \exists t' \geq t \quad \text{such that} 
\quad \Delta S_b \subseteq M_a(t') \subseteq \pi_{O_a}(\mathcal{M}(t')).
\end{equation}
This lemma assumes reliable refresh propagation and does not apply under probabilistic delivery, which is handled separately in Theorem~\ref{thm:dynamic-convergence}.
\end{lemma}

\begin{proof}[Proof Sketch.]
Since $O_a \cap \mathrm{dom}(\Delta S_b) \neq \emptyset$, a refresh notification is eventually delivered to agent $a$ by design. Upon receiving it, $a$ retrieves and integrates $\Delta S_b$ after local validation, ensuring $M_a(t') \subseteq \pi_{O_a}(\mathcal{M}(t')$ for some finite $t' \geq t$.
\end{proof}

\begin{theorem}[Monotonic Semantic Coherence]\label{thm:semantic-coherence}
If all structured updates are validated against the ontology $\mathcal{O}$ prior to integration, then at all times $t$, the global memory $\mathcal{M}(t)$ remains consistent with $\mathcal{O}$.

Formally:
\begin{equation}
\forall t, \quad \mathcal{O} \vdash \mathcal{M}(t)
\end{equation}
\end{theorem}

\begin{proof}[Proof Sketch.]
Since $\mathcal{M}(0)$ is $\mathcal{O}$-consistent and each integrated update $\Delta S_a^{\text{valid}}$ satisfies $\mathcal{O} \vdash \Delta S_a^{\text{valid}}$, consistency is preserved inductively: $\mathcal{O} \vdash \mathcal{M}(t)$ for all $t$.
\end{proof}

\begin{corollary}[Slice Validity Preservation]
If $\mathcal{M}(t')$ is consistent with $\mathcal{O}$ for all $t' \leq t$, then every agent’s semantic slice $S_a(t)$ is consistent with $\mathcal{O}$.

Formally:
\begin{equation}
\forall a \in \mathcal{A}, \quad \exists t' \leq t \quad \text{such that} \quad M_a(t) \subseteq \pi_{O_a}(\mathcal{M}(t')) \quad \text{and} \quad \mathcal{O} \vdash S_a(t)
\end{equation}
\end{corollary}

\begin{proof}[Proof Sketch]
By Theorem~\ref{thm:semantic-coherence}, $\mathcal{M}(t')$ is $\mathcal{O}$-consistent for all $t' \leq t$. Since $M_a(t) \subseteq \pi_{O_a}(\mathcal{M}(t'))$ and projection preserves $\mathcal{O}$-consistency over $O_a \subseteq \mathcal{O}$, we have $\mathcal{O} \vdash S_a(t)$. 
\end{proof}

\subsubsection{Fault Tolerance and Containment}\label{subsubsec:fault-tolerance}

We next examine how the system behaves under partial failure, showing that agent-level faults are contained locally and do not compromise global consistency.

\begin{proposition}[Localized Failure Containment]\label{prop:local-failure}
If an agent $a \in \mathcal{A}$ fails at time $t_f$, only its semantic slice $S_a(t_f)$ and pending updates are lost. The global memory $\mathcal{M}(t)$ and the slices of all other agents remain consistent with $\mathcal{O}$.

Formally:
\begin{equation}
\forall b \in \mathcal{A} \setminus \{a\}, \quad \mathcal{O} \vdash S_b(t) \quad \text{for all } t \geq t_f
\end{equation}
\end{proposition}

\begin{proof}[Proof Sketch]
Since $\mathcal{M}(t)$ is composed of validated updates, failure of agent $a$ affects only its local slice $S_a(t)$ and any pending proposals. The global memory $\mathcal{M}(t)$ and unaffected agents' slices remain $\mathcal{O}$-consistent.
\end{proof}

\subsubsection{Convergence and Synchronization}\label{subsubsec:convergence}

We now show that, in the absence of new updates and under reliable communication, agents eventually synchronize their semantic slices with the stable global memory.

\begin{corollary}[Quiescence under No New Updates]\label{cor:refresh-completes}
Given lemma~\ref{lem:delivery} and assuming no new updates occur after time $t_0$, all refresh notifications and necessary agent-to-agent retrievals complete within finite time.

Formally:
\begin{equation}
\exists T > t_0 \quad \text{such that} \quad \forall a \in \mathcal{A}, \quad 
\text{all } \tau(\Delta S_a) \text{ and agent\text{-}to\text{-}agent retrievals complete by } T
\end{equation}
\end{corollary}

\begin{proof}[Proof Sketch]
Since no new updates occur after $t_0$, only finitely many refresh notifications are issued. By Lemma~\ref{lem:delivery}, each relevant notification and retrieval completes within finite time, yielding global communication quiescence by some $T > t_0$. 
\end{proof}

\begin{theorem}[Slice Convergence]\label{thm:slice-consistency}
Assume Lemma~\ref{lem:delivery} and that no new updates are committed after time $t_0$.
Then for every agent $a \in A$, there exists a finite time $T' \geq t_0$ such that
\begin{equation}
M_a(T') = \pi_{O_a}\bigl(\mathcal{M}(t_0)\bigr)
\quad\text{and}\quad
M_a(t) = M_a(T') \quad \text{for all } t \geq T'.
\end{equation}
\end{theorem}

\begin{proof}[Proof Sketch.]
Since no new updates occur after $t_0$, the global memory stabilizes to $\mathcal{M}(t_0)$.  
By Lemma~\ref{lem:delivery} and Corollary~\ref{cor:refresh-completes}, all relevant refresh notifications and retrievals complete within finite time.  
Assuming agents integrate received updates without loss or reordering,  
each local memory $M_a(t)$ converges to $\pi_{O_a}\bigl(\mathcal{M}(t_0)\bigr)$ by some finite time $T'$.  
Thereafter, no further changes occur.
\end{proof}

\subsubsection{Ontology-Scoped Causal Isolation}\label{subsubsec:causal- isolation}

Semantic Fusion guarantees that agents are unaffected by updates involving entities outside their scoped ontology. This yields strong modularity and reduces unnecessary memory turnover.

\begin{theorem}[Ontology-Scoped Causal Isolation]\label{thm:causal-isolation}
If \( \text{dom}(s) \notin O_a \) for all \( s \in \Delta S_b \), then:
\begin{equation}
\tau(\Delta S_b) \cap O_a = \varnothing \quad \Rightarrow \quad M_a(t+1) = M_a(t)
\end{equation}
\end{theorem}

\begin{proof}[Proof Sketch]
If \(\mathrm{dom}(\Delta S_b) \cap O_a = \emptyset\), no refresh notification is received by agent \(a\), so \(M_a(t+1) = M_a(t)\).
\end{proof}

\subsubsection{Communication Complexity}\label{subsubsec:comms-complexity}

Finally, we analyze the communication overhead associated with update propagation, showing that the system achieves bounded complexity relative to the number of impacted agents.

\begin{proposition}[Bounded Communication Complexity]\label{prop:comm-bound}
Assume that refresh notifications $\tau(\Delta S_a)$ are disseminated with negligible overhead compared to agent-to-agent retrievals. Let $d$ be the number of agents whose semantic slices intersect the updated entities. Then each validated update incurs $O(d)$ communication overhead, independent of the total number of agents.

Formally:
\begin{equation}
\text{Communication Cost} = O(d)
\end{equation}
\end{proposition}

\begin{proof}[Proof Sketch]
Each validated update \(\Delta S_a^{\text{valid}}\) triggers lightweight refresh notifications identifying the modified entities. Only agents with slice overlap (\(O_b \cap \mathrm{dom}(\Delta S_a) \neq \varnothing\)) perform retrievals. Since retrievals dominate communication cost, and occur for at most \(d\) agents, total communication is \(O(d)\) per update. (e.g., $O(n)$ broadcast still yields $O(d)$ dominant cost under scoped pulls).
\end{proof}

\begin{theorem}[Lower Bound on Communication under Scoped Validation]\label{thm:comm_lower_bound}
Let an update $\delta \in\Delta S_a$ affect a set of semantic entities $\mathcal{E} \subseteq \mathcal{U}$, and let $d$ be the number of agents whose semantic slices intersect with $\mathcal{E}$. Then in the worst case, any protocol that preserves ontology-scoped validity must incur at least $\Omega(d)$ communication events.
\end{theorem}

\begin{proof}[Proof Sketch]
Let \(\mathcal{E}\) be a set of entities such that each of the \(d\) agents has a nonempty intersection with \(\mathcal{E}\).  
To preserve slice validity, each agent with \(O_b \cap \mathcal{E} \neq \varnothing\) must observe the update \(\delta \in \Delta S_a\).  
If fewer than \(d\) agents receive or retrieve \(\delta\), some slice will fail to integrate a valid scoped update, violating Lemma~\ref{lem:delivery}.  
Thus, \(\Omega(d)\) communication is necessary.
\end{proof}

\begin{corollary}[Tightness of Upper Bound]\label{cor:bound-tightness}
The $O(d)$ communication bound established in Proposition~\ref{prop:comm-bound} is asymptotically tight. The system achieves the minimum message complexity consistent with ontology-scoped semantic coherence.
\end{corollary}

\begin{proof}[Proof Sketch]
Proposition~\ref{prop:comm-bound} shows $O(d)$ upper bound per update, and Theorem~\ref{thm:comm_lower_bound} shows matching $\Omega(d)$ lower bound.  
Thus, message complexity is $\Theta(d)$, and the bound is asymptotically tight.
\end{proof}

\begin{remark}
Each relevant agent must observe the update to validate or apply it, so $\Theta(d)$ communication is unavoidable in the worst case. Implementations can achieve this bound using lightweight, non-duplicating notifications followed by scoped pulls.
\end{remark}

\paragraph{Summary.}
Together, Theorems~\ref{thm:semantic-coherence} through~\ref{thm:slice-consistency} establish that under reliable delivery and ontology-conformant updates, each agent’s slice remains semantically coherent, aligned with relevant changes, and eventually converges to a consistent view of the global memory, without requiring centralized coordination or full synchronization.

Nonetheless, decentralization does not imply the elimination of all information flow, but rather the removal of centralized coordination and direct inter-agent negotiation. In Semantic Fusion, semantic coherence is preserved through structured, ontology-scoped refresh notifications, which incur only $\Theta(d)$ communication overhead---the minimum necessary for correctness (Theorem~\ref{cor:bound-tightness}). Without such bounded, scoped propagation, no decentralized system could maintain coherent shared semantic memory.

\subsection{Behavioral Equivalence}\label{subsec:behavioral-equivalence}
We now formalize the relationship between an agent’s local memory evolution and the evolution of its projected global state over time. To do so, we model both the global memory and each agent’s local memory as labeled transition systems (LTSs), which capture the sequence of update-induced state changes.

\begin{definition}[Projected Global LTS]
Let $\mathcal{G} = (Q, \rightarrow)$ be the labeled transition system over global memory states $\mathcal{M}(t)$, where each transition corresponds to the integration of a validated update $\Delta S_a^{\text{valid}}$. The projection $\mathcal{G}\!\mid_{O_a}$ is the LTS over projected states $\pi_{O_a}(\mathcal{M}(t))$, with transitions restricted to updates involving entities in $O_a$.
\end{definition}

\begin{definition}[Agent Execution LTS]
Let $\mathcal{T}_a = (Q_a, \Rightarrow)$ be the transition system over local memory states $M_a(t)$ for agent $a$, where transitions arise from (i) locally generated validated proposals and (ii) integration of relevant refresh-triggered updates.
\end{definition}

\begin{theorem}[Slice-Global Semantic Bisimulation]\label{thm:slice_bisim}
By Lemma~\ref{lem:delivery} and Assumption S3 (see Appendix~\ref{app:proof-thm10}), there exists a stuttering bisimulation relation $\mathcal{R}_a \subseteq Q_a \times \pi_{O_a}(Q)$ such that:
\begin{equation}
\exists\,t' \le t \;:\;
\bigl(M_a(t),\;\pi_{O_a}\!\bigl(M(t')\bigr)\bigr)\; \in R_a .
\end{equation}

That is, the agent’s local execution is stuttering bisimilar to its scoped projection of the global execution:
\begin{equation}
\mathcal{T}_a \;\dot{\sim}_{\text{st}}\; \mathcal{G}\!\mid_{O_a}.
\end{equation}
\end{theorem}

\begin{proof}[Proof Sketch]
Initially, $M_a(0) = \pi_{O_a}(\mathcal{M}(0))$.  
Each local transition is either a validated proposal (reflected globally) or integration of a scoped update already committed to $\mathcal{M}(t')$ for some $t' \leq t$.  
By Lemma~\ref{lem:delivery}, all slice-relevant updates eventually propagate, ensuring $M_a(t)$ remains stuttering bisimilar to $\pi_{O_a}(\mathcal{M}(t'))$.  
Updates outside $O_a$ are ignored locally, preserving alignment.  
Forth and back conditions hold; see Appendix~\ref{app:proof-thm10} for full construction.
\end{proof}

\begin{remark}
This result establishes the first formal bridge from local semantic slices to the projected global memory, showing that decentralized agents can reason over their scoped memory without violating global correctness. It implies that the union of local slices, each bisimilar to its respective global projection, collectively captures the behavior of the full system. The global memory need not exist as a central structure: its semantics are preserved in the distributed ensemble of slice-local executions.
It enables modular verification, local model checking, and runtime monitoring with global semantic alignment.
\end{remark}

\begin{corollary}[Temporal Property Preservation]\label{cor:ltl_preservation}
Let $\varphi$ be any temporal logic formula (LTL or CTL*) over entities in $O_a$. If the projected global execution satisfies $\varphi$, i.e., $\mathcal{G}|_{O_a} \models \varphi$, then the agent's local execution satisfies it as well:
\begin{equation}
\mathcal{T}_a \models \varphi.
\end{equation}

That is, semantic properties verified over the global projection are preserved under local execution.\footnote{For CTL$^\ast$, stuttering preservation holds for all formulas
that do not use the next‑state operator $X$; see
\cite{kupferman2000modelcheck} for details.}
\end{corollary}

\begin{proof}[Proof Sketch]
By Lemma~\ref{lem:delivery} and Theorem~\ref{thm:slice_bisim}, projected global transitions over $O_a$ eventually appear in $M_a(t)$.  
Stuttering bisimulation preserves satisfaction of all LTL and CTL*$\setminus$\{X\} formulas~\cite{kupferman2000modelcheck}, so if $\mathcal{G}|_{O_a} \vDash \varphi$, then $\mathcal{T}_a \vDash \varphi$.
\end{proof}

\begin{theorem}[Slice–Local Completeness for Safety Properties]
\label{thm:slice-completeness-safety}
Let $\varphi$ be any safety property expressible in temporal logic over semantic entities in $\mathcal{O}$. Suppose each agent $a \in A$ operates over an ontology slice $O_a \subseteq \mathcal{O}$, and the family $\{O_a\}_{a \in A}$ covers $\mathcal{O}$. If every agent's execution trace satisfies $\varphi$ locally, then the global execution over $\mathcal{M}(t)$ satisfies $\varphi$ as well.

\begin{equation}
\forall a \in A,\ \mathcal{T}_a \models \varphi \quad \Rightarrow \quad \mathcal{G} \models \varphi
\end{equation}
\label{thm:slice-completeness-safety}
\end{theorem}

\begin{proof}[Proof Sketch]
Suppose, for contradiction, that the global execution $\mathcal{G}$ violates the safety property $\varphi$.  
Then there exists a finite bad prefix $\pi = \mathcal{G}(0), \ldots, \mathcal{G}(k)$ that violates $\varphi$.  
Since $\{O_a\}_{a\in A}$ covers $\mathcal{O}$, some agent $a^*$ has a slice $O_{a^*}$ intersecting $\pi$.  
Projecting $\pi$ onto $O_{a^*}$ yields a trace that still violates $\varphi$.  
By Theorem~\ref{thm:slice_bisim}, $\mathcal{T}_{a^*}$ is stuttering bisimilar to this projection, so $\mathcal{T}_{a^*} \nvDash \varphi$, contradicting the premise.  
Thus, $\mathcal{G} \vDash \varphi$. Full proof in \ref{app:proof-slice-completeness}.
\end{proof}

\paragraph{Summary.}
The bisimulation result (Theorem~\ref{thm:slice_bisim}) formally connects local slice memory to projected global behavior. Temporal logic formulas (e.g., CTL*, LTL) that hold on the global trace also hold locally, enabling modular verification and decentralized runtime monitoring. Slice-local Completeness (Theorem~\ref{thm:slice-completeness-safety}) establishes that global safety can be verified slice-by-slice: if every agent satisfies the property over its scoped projection, the entire system satisfies it globally. The full proof appears in Appendix~\ref{app:proof-slice-completeness}.

\subsection{Probabilistic Guarantees}\label{subsec:prob-guarantees}
We now extend the previous results to agents whose update proposals may be sampled from distributions conditioned on their current slices, such as agents driven by probabilistic planners or incorporating LLM-based components. Each agent proposes updates by sampling from a distribution conditioned on its current slice, and transitions are represented using probabilistic labeled transition systems (pLTSs). We show that local and projected global executions remain probabilistically bisimilar under slice-scoped validation and reliable refresh.

We begin by defining a probabilistic generalization of the update process and execution model introduced earlier.

\begin{definition}[Probabilistic Update Generator]
Let $\gamma_a : S_a(t) \rightarrow \mathsf{Dist}(\candP{a})$ be a probabilistic
update generator for agent $a$, where $\candP{a} \subseteq \Delta S_a$ is the agent's set
of candidate proposals. Given the current slice $S_a(t)$, the generator returns a
discrete distribution over $\candP{a}$. Invalid proposals are filtered by ontology
validation, yielding $\Delta S_a^{\text{valid}}$.
\end{definition}

\begin{definition}[Probabilistic LTS]
A probabilistic labeled transition system (pLTS) is a tuple $\mathcal{P} = (Q, \xrightarrow{\alpha}, \mathcal{L}_\mathsf{p})$, where:
$Q$ is a set of states,
$\xrightarrow{\alpha} \subseteq Q \times \mathsf{Dist}(Q)$ is a labeled probabilistic transition relation,
$\mathcal{L}_\mathsf{p} : Q \rightarrow \mathrm{Val}$ is a labeling function assigning semantic truth valuations. Here, Val denotes the set of possible mappings from ontology entities to their current values or truth states.
The transition label $\alpha$ may denote different types of activity, such as update proposals (\textsc{prop}) or refresh-triggered state changes (\textsc{sync}).
\end{definition}

\begin{definition}[Projected Global pLTS]
Let $\mathcal{P}_G = (Q^\mathsf{p}, \xrightarrow{}, \mathcal{L}_G^\mathsf{p})$ denote the probabilistic labeled transition system (pLTS) over global memory states $\mathcal{M}(t)$, where each transition corresponds to the integration of a validated update $\Delta S_b^{\text{valid}}$. Here, $\mathcal{L}_G^\mathsf{p} : Q^\mathsf{p} \to \mathsf{Val}$ is the labeling function that assigns semantic valuations over ontology entities. 

The projection $\pLTS{G}\!\mid_{O_a}
$ is the pLTS over the projected memory $\pi_{O_a}(\mathcal{M}(t))$, with transitions restricted to updates involving entities within the slice $O_a$.
\end{definition}

\begin{definition}[Agent Execution pLTS]
Let $\mathcal{P}_a = (Q_a^\mathsf{p}, \xrightarrow{}, \mathcal{L}_a^\mathsf{p})$ denote the probabilistic labeled transition system (pLTS) over local memory states $\mathcal{M}_a(t)$ for agent $a$. 
Here, $\mathcal{L}_a^\mathsf{p} : Q_a^\mathsf{p} \to \mathsf{Val}$ is the labeling function that assigns semantic valuations over entities in the agent's ontology slice $O_a$.

Transitions in $\mathcal{P}_a$ arise from:
i) Probabilistically sampled update proposals drawn from the agent’s update generator $\gamma_a$ (\textsc{prop}), and
ii) Integration of externally validated updates received through scoped refresh propagation (\textsc{sync}).
\end{definition}

\begin{definition}[Probabilistic Bisimulation]
A relation $\mathcal{R}_a^\mathsf{p} \subseteq Q_a^\mathsf{p} \times Q^\mathsf{p}$ is a probabilistic bisimulation if, whenever $(q_a, q_G) \in \mathcal{R}_a^\mathsf{p}$:
$\mathcal{L}_a^\mathsf{p}(q_a) = \mathcal{L}_G^\mathsf{p}(q_G)$, and
for every transition $q_a \xrightarrow{\alpha} \mu_a$, there exists $q_G \xrightarrow{\alpha} \mu_G$ such that $\mu_a$ and $\mu_G$ admit a coupling with support in $\mathcal{R}_a^\mathsf{p}$, and vice versa \cite{segala1995probbisim}.

\begin{remark}
    In the case of validated refresh updates, the transition is deterministic. We denote this using a Dirac delta distribution \( \delta_{x} \), representing a probability mass concentrated at state \( x \).
\end{remark}
\end{definition}

\begin{theorem}[Probabilistic Slice--Global Bisimulation]\label{thm:prob_slice_bisim}
Under slice-conditioned update generation, Assumption S3 (Appendix~\ref{app:proof-thm10}), and Lemma~\ref{lem:delivery}, there exists a probabilistic bisimulation relation \(\mathcal{R}^p_a \subseteq Q^p_a \times Q^p\) such that:
\begin{equation}
\exists\,t' \le t \;:\;
\bigl(M_a(t),\;\pi_{O_a}\!\bigl(M(t')\bigr)\bigr)\; \in \mathcal{R}^p_a.
\end{equation}

That is, the agent’s probabilistic local execution is bisimilar to the projection of the global probabilistic execution over its ontology slice:
\begin{equation}
\mathcal{P}_a \;\sim_{\text{pbis}}\; \mathcal{P}_G|_{O_a}.
\end{equation}
\end{theorem}

\begin{proof}[Proof Sketch]
Each proposal $u \in P_a$ sampled by $\gamma_a(S_a(t))$ with probability $p_u$ yields a local transition with matching probability.  
If validated, $u$ is integrated into $\mathcal{M}(t')$ and reflected in $\pi_{O_a}(\mathcal{M}(t'))$.  
Refresh propagation (Lemma~\ref{lem:delivery}) ensures that external updates synchronize locally.  
A probabilistic coupling, $\kappa$,  between transitions maintains bisimulation; see~\cite{segala1995probbisim} for full construction. Full proof in Appendix~\ref{app:proof-thm12}.
\end{proof}

\begin{remark}
This result generalizes slice-global semantic bisimulation to the case where update proposals may vary across invocations. It enables compositional verification and slice-local monitoring in hybrid systems with probabilistic reasoning. This relation is visualized in Figure \ref{fig5}, which illustrates the alignment of local and global transitions under the bisimulation relation~\cite{segala1995probbisim}.
\end{remark}

\begin{figure}[H]
    \centering
    \includegraphics[width=0.85\textwidth]{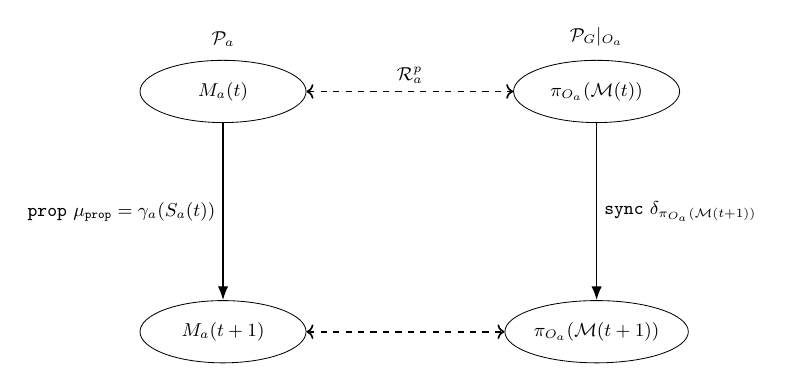}
    \caption{Probabilistic bisimulation between an agent’s local memory execution and its projected global memory slice under probabilistic update generation and refresh synchronization. The sync transition is depicted as aligning the agent’s memory with the global projection at time~\(t{+}1\) for illustrative clarity. In general, however, an agent’s local memory may lag, and the corresponding global state reflected in its slice may have been committed at an earlier time~\(t' \le t{+}1\), consistent with Theorem~\ref{thm:prob_slice_bisim}}
    \label{fig5}
\end{figure}

\begin{corollary}[Probabilistic Temporal Property Preservation]\label{cor:pctl_preservation}
Let $\varphi$ be any PCTL* formula over entities in $O_a$ that does not use the next-state operator $X$.
If the projected global execution satisfies $\Pr_{\ge p}[\varphi]$, then so does the local agent execution:
\begin{equation}
\mathcal{P}_G|_{O_a} \models \Pr_{\ge p}[\varphi]
\quad \Longrightarrow \quad
\pLTS{a} \models \Pr_{\ge p}[\varphi].
\end{equation}
\end{corollary}

\begin{proof}[Proof Sketch]
By Theorem~\ref{thm:prob_slice_bisim}, $\pLTS{a} $ and $\mathcal{P}_G|_{O_a}$ are related by probabilistic bisimulation.  
Probabilistic bisimulation preserves satisfaction of PCTL*$\setminus$\{X\} formulas~\cite{baier2008principles}, so the satisfaction probability of $\varphi$ matches between local and global executions.
\end{proof}

\begin{corollary}[Probabilistic Slice--Local Completeness]\label{cor:prob_slice_local_completeness}
Let $\varphi$ be any safety property expressible over ontology $O$, and suppose the family $\{O_a\}_{a\in A}$ covers $\mathcal{O}$.
If every agent execution $\pLTS{a}$ satisfies $\Pr_{\ge p}[\varphi]$, then the global execution $\pLTS{G}$ satisfies $\Pr_{\ge p}[\varphi]$ as well:
\begin{equation}
\forall a\in A,\quad \pLTS{a} \models \Pr_{\ge p}[\varphi]
\quad \Longrightarrow \quad
\pLTS{G} \models \Pr_{\ge p}[\varphi].
\end{equation}
\end{corollary}

\begin{proof}[Proof Sketch]
Suppose, for contradiction, that $\mathcal{P}_G$ violates $\varphi$ with probability exceeding $1-p$.  
By slice coverage and Theorem~\ref{thm:prob_slice_bisim}, some agent's execution $\pLTS{a} $ must also reflect a violation with probability exceeding $1-p$, contradicting the premise.  
Thus, $\mathcal{P}_G \vDash \text{Pr}_{\geq p}[\varphi]$.
\end{proof}

\subsubsection*{Probabilistic Validation (Optional Extension)}
\label{subsubsec:prob-validation}

In practical deployments, updates may originate from modules whose outputs
exhibit variability or uncertainty (for example, some learned and
data-driven models), which can generate proposals that are syntactically
correct but semantically invalid. To reflect this, we introduce a notion
of \emph{probabilistic validation}. Let $\theta$ denote a proposed
structured update (which may be valid or invalid).

\begin{definition}[Probabilistic Local Validator]
Let $\varepsilon_a$ denote the probability that agent $a$ incorrectly
accepts an invalid proposal $\theta$ (false acceptance).  
Let $\xi_a$ denote the probability that agent $a$ incorrectly rejects
a valid proposal (false rejection). Then:
\[
\Pr[\text{accept } \theta \mid \theta \text{ invalid}] = \varepsilon_a,
\qquad
\Pr[\text{reject } \theta \mid \theta \text{ valid}] = \xi_a.
\]
\end{definition}

If each agent $a$ has $\varepsilon_a \ll 1$, and update proposals are
independently validated by all agents whose ontology slices intersect the
affected entities, then the probability that an \emph{invalid} update is
incorrectly reflected in global memory $\mathcal{M}(t)$ becomes
exponentially small in the number of such overlapping validators.

\begin{theorem}[Probabilistic Global Coherence]
\label{thm:prob-coherence}
Suppose each invalid proposal is accepted locally with probability at most
$\varepsilon_{\max}$, and must be independently validated by $r$ agents
(including the proposer). Then the probability that an invalid proposal is
committed to $\mathcal{M}(t)$ is at most $(\varepsilon_{\max})^{r}$:
\begin{equation}
\Pr[\theta \text{ invalid and committed to } \mathcal{M}(t)]
\le (\varepsilon_{\max})^{r}.
\end{equation}
\end{theorem}

\begin{proof}[Proof Sketch.]
Let $\theta$ be an invalid update proposed by agent $a$. Each of the $r$
validators---the originator and all agents whose slices intersect its
scope---accepts it with probability at most $\varepsilon_{\max}$. The
validators act independently, so the joint probability is at most
$(\varepsilon_{\max})^{r}$.
\end{proof}

A false rejection only harms the rejecting agent: if an agent incorrectly rejects a
valid update (a false rejection), it performs a semantic stutter step
$M_b(t+1)=M_b(t)$, which constitutes a locally absorbing non-update. Such
local rejections do not modify global memory and do not affect the
probabilistic consistency guarantees of Theorem~\ref{thm:prob-coherence}.

\begin{remark}
This result extends monotonic semantic coherence
(Theorem~\ref{thm:semantic-coherence}) to settings with imperfect
validators. Even without centralized coordination, the system remains
resilient: an invalid update must be simultaneously accepted by all
validators to corrupt global memory. The probability of total semantic
corruption decays exponentially with the number of validators. This becomes
particularly relevant when agents rely on probabilistic or learned modules,
including LLM-based components, whose outputs are not guaranteed to satisfy
all semantic constraints in every instance.
\end{remark}

\begin{example}
\textbf{Illustration.}
Suppose $\varepsilon_{\max} = 0.02$ and $r = 3$ overlapping validators.
Then the probability of a globally committed invalid update is at most
$(0.02)^3 = 8 \times 10^{-6}$.
\end{example}

\paragraph{Summary.}
Theorems~\ref{thm:prob_slice_bisim} and~\ref{thm:prob-coherence} provide
probabilistic guarantees on semantic coherence: independent slice-local
validators reduce the probability of an invalid update entering shared
memory, ensuring reliable state evolution even when proposal generation is
nondeterministic.
\subsection{Dynamic Semantic Convergence}\label{subsec:dynamic-convergence}

While the preceding results assumed a fixed set of validated updates committed by a finite time $t_0$, in practical settings semantic memory evolves continuously. New updates may be integrated indefinitely as agents observe and interact with their environments.  
The following theorem formalizes \emph{dynamic semantic convergence}: under mild probabilistic assumptions, agents re-synchronize their memory slices with the evolving global state even without retries of missed earlier updates.

\begin{definition}[Stream-Fair Coverage]\label{def:stream-fair-coverage}
Let $\{\delta^{(t)}\}_{t\ge t_0} \subseteq \Delta$ denote the (potentially infinite) stream of validated updates committed after time $t_0$.  
For an agent $b$ with ontology slice $O_b$, define the event
\[
C_b(t) = \bigl( O_b \cap \mathrm{dom}\bigl(\delta^{(t)}\bigr) \neq \varnothing),
\]
indicating that the $t^\text{th}$ update intersects $b$'s slice.  
The stream satisfies \emph{stream-fair coverage} with parameter $\eta > 0$ if
\[
\Pr\bigl[C_b(t)\bigr] \ge \eta
\quad\text{independently for all } t\text{ and each agent } b.
\]
\end{definition}

\begin{theorem}[Dynamic Semantic Convergence]
\label{thm:dynamic-convergence}
Fix an agent \( b \) with ontology slice \(\ O_b \), and a (potentially infinite) stream of validated updates \( \{ \delta^{(t)} \}_{t \ge t_0} \).  
For each update \( \delta^{(t)} \), let \( \mathcal{E}^{(t)} \subseteq \mathcal{U} \) denote the set of modified semantic entities (i.e., the refresh notification target as in Definition~\ref{def:refresh}).

Suppose that:
\begin{itemize}
\item[(i)] Every update is ontology-consistent and scoped to a finite set of ontology entities \( \mathcal{E}^{(t)} \);
\item[(ii)] For each update \( \delta^{(t)} \), if \( O_b \cap \mathcal{E}^{(t)} \neq \emptyset \), the scoped refresh is delivered to agent \( b \) with probability at least \( \rho > 0 \), independently across \( t \);
\item[(iii)] For each entity \( x \in O_b \), the probability that \( x \in \mathcal{E}^{(t)} \) at a given time step is at least \( \eta > 0 \), independently across \( t \).
\end{itemize}

Then, with probability 1, for every entity \( x \in O_b \),
\begin{equation}
\lim_{t \to \infty} \mathcal{M}_b(t)[x] = \mathcal{M}(t)[x].
\end{equation}
Moreover, the time required for agent \( b \) to receive \( r \) relevant updates grows with high probability no faster than \( \Theta(r / (\rho \eta)) \), and the probability of delay beyond this bound decays exponentially in \( r \).

\end{theorem}

\begin{proof}[Proof Sketch]
Let $\mathcal{T} = \{ t \ge t_0 : O_b \cap \mathrm{dom}(\delta^{(t)}) \neq \varnothing \}$ denote the time steps when updates touch $b$'s slice.  
By assumption (iii), the expected density of $\mathcal{T}$ is at least $\eta$.

For each $t \in \mathcal{T}$, define $Y_t = 1$ if the scoped refresh for $\delta^{(t)}$ is successfully delivered to $b$, and $0$ otherwise.  
By assumption (ii), the $Y_t$ are independent and $\Pr[Y_t = 1] \ge \rho$.

Let $Z_r$ be the time by which $b$ has received $r$ slice-relevant updates:
\[
Z_r = \min \left\{ t : \sum_{t' \in \mathcal{T}, t' \le t} Y_{t'} \ge r \right\}.
\]
By a Chernoff bound, we have:
\[
\Pr[ Z_r - t_0 > k ] \le e^{-\Omega(r)} \quad \text{for } k = \Theta(r / (\rho \eta)).
\]

Each successful refresh integrates an ontology-valid update into $M_b(t)$.  
Because each key in $O_b$ is touched with nonzero frequency, and each such update has a nonzero chance of being delivered, every key is eventually refreshed.  
Thus, with probability $1$, $M_b(t)$ converges to the projected global memory:
\[
\lim_{t \to \infty} M_b(t) = \pi_{O_b}(\mathcal{M}(t)).
\]
\end{proof}

\noindent\emph{Note.} Lemma~\ref{lem:delivery} assumes reliable delivery and is used in the bisimulation setting (Theorem~\ref{thm:slice_bisim}).  
Theorem~\ref{thm:dynamic-convergence} does not assume universal delivery; instead, it guarantees that relevant updates are received and integrated with high probability over time.

\noindent\textbf{Communication Cost.}  
Each validated update incurs exactly one $\Theta(d)$ scoped refresh attempt at commitment, independent of retries.  
Thus, the total communication overhead per update remains $\Theta(d)$, preserving the scaling bound established in Corollary~\ref{cor:bound-tightness}.

\section{Experimental Validation}\label{sec:experimental-val}

To validate the Semantic Fusion (SF) model under operational constraints, we implemented a lightweight instantiation that exercises its core constructs in a dynamic simulation environment. The goal was not empirical benchmarking, but structured confirmation that the model’s theoretical properties hold under interpretable, asynchronous conditions. The experiment evaluates the guarantees developed in the formal model—including semantic coherence (Theorem~\ref{thm:semantic-coherence}), probabilistic slice convergence (Theorem~\ref{thm:dynamic-convergence}), causal isolation (Theorem~\ref{thm:causal-isolation}), stuttering bisimulation (Theorem~\ref{thm:slice_bisim}), probabilistic bisimulation (Theorem~\ref{thm:prob_slice_bisim}), and local failure containment (Proposition~\ref{prop:local-failure}).

A second purpose of the experiment is to demonstrate that coordinated behavior can arise from interaction with shared semantic memory. The experiment therefore tests whether agents, despite having partial and asynchronous local views, can maintain an aligned and semantically coherent common operating picture through the shared semantic state.

\subsection{Simulation Scenario}

The scenario used 250 autonomous agents: 50 of type \textit{search} and 100 each of type \textit{relay} and \textit{rescue}, collaborating in a decentralized search-and-rescue mission via ontology-governed memory. Search agents monitored assigned regions and proposed updates when their internal policy generated stochastic survivor indications. Relay agents repositioned in response to active zones and logged their presence on arrival. Rescue agents listened for semantic cues, bid for recovery tasks, and executed actions only when relay support was available. All interaction occurred through semantic memory updates, without centralized coordination.

Each agent’s observation space consists solely of (i) its ontology-scoped semantic slice $S_a(t)$ and (ii) stochastic triggers from its internal scripted policy. Agents have no access to exogenous variables or hidden world structures. The action space consists of proposing updates, moving between adjacent zones, and issuing task bids (for rescue agents). All agents execute as independent asynchronous coroutines; the global tick advances only after all per-agent coroutines complete.

Midway through the run, two agents were removed without prior signaling to evaluate fault resilience.

\subsection{Simulation Setup}
\label{subsec:sim-setup}

Each agent operated over a scoped semantic slice derived from a shared ontology
encoding task-relevant entities such as search zones, survivor statuses, and
relay coverage. Validated updates were disseminated through a push-based
scoped-delivery mechanism: only agents whose slices intersected a modified
entity received refresh notifications, after which they revalidated and
integrated the update. This routing layer enforced slice-local visibility and
exercised the same semantics as the model’s selective refresh process without
assuming a particular communication substrate.

The environment was fully symbolic. Zones were grid identifiers, survivors were
logical state variables, and agents interacted only with symbolic percepts
rather than a continuous or embodied world. Each agent had a defined
observation space: search agents observed survivor state in their current zone,
relay agents observed coverage gaps, and rescue agents observed task
announcements. Observations were converted into
structured update proposals during asynchronous agent steps.

Agent behavior followed lightweight scripted policies chosen to produce
interpretable update streams. Search agents moved deterministically but sampled
survivor detections with a fixed $30\%/70\%$ distribution. Relay and rescue
agents were initialized at random positions; rescuers generated bids with
random tie-breaking and sampled service durations from a fixed distribution.
These policies introduced heterogeneous stochastic proposals without learned
controllers or external perception modules, ensuring that observable effects
reflected the SF model rather than complex agent decision logic. All
randomness originated in the policy modules rather than the environment.

To confirm that ontology validation was non-vacuous, we injected three
adversarial proposals that violated ontology constraints. Each was rejected
locally through the standard validation path, producing a stutter step without
altering global memory.

These scripted policies were intentionally simple: their purpose was to
exercise the SF semantics under interpretable conditions, not to approximate
realistic sensing or planning. Evaluating SF under richer or degraded policies
is an important direction for future work.

\subsection{Evaluation Criteria}
\label{subsec:eval-criteria}

Simulation logs were analyzed to verify each of the following theoretical properties:

\begin{itemize}
  \item \textbf{Semantic Coherence (Theorem~\ref{thm:semantic-coherence}):} Number of invalid updates accepted by any agent.
  \item \textbf{Causal Isolation (Theorem~\ref{thm:causal-isolation}):} Whether all updates integrated into each agent’s memory fell strictly within its semantic slice.
  \item \textbf{Stuttering Bisimulation (Theorem~\ref{thm:slice_bisim}):} Forward temporal alignment between each agent’s memory trace and the projected global execution over its slice. The reverse direction is unnecessary for the safety guarantees we claim.
  \item \textbf{Probabilistic Bisimulation (Theorem~\ref{thm:prob_slice_bisim}):} Whether validated probabilistic proposals were eventually reflected in scoped memory.
  \item \textbf{Failure Containment (Proposition~\ref{prop:local-failure}):} Post-removal continuation of valid update propagation and maintenance of memory consistency.
\end{itemize}

\subsection{Results}
\label{subsec:results}

Each of the formal properties was upheld in the simulation, which processed 11,325 updates. In the main simulation runs, no invalid proposals were generated by the scripted policies. All invalid-update rejection behavior therefore arises exclusively from the adversarial injection mechanism described in \ref{subsec:sim-setup}. No invalid updates were accepted 
(Theorem~\ref{thm:semantic-coherence}) and every validated update conformed to the agent’s scoped ontology (Theorem~\ref{thm:causal-isolation}).

Each agent’s local memory remained aligned with its projected global slice under bounded delay\footnote{For evaluation purposes, we exclude the final time step from bisimulation tests, as agents are not given a chance to re-align after the last committed updates.} (Theorem~\ref{thm:slice_bisim}), and 1,250 validated probabilistic proposals\footnote{Bisimulation tests are based on snapshot alignment between local memory and the agent's projected global slice, not on event log volume. The validation compares memory states extracted at each time step, allowing bounded delay, and verifies that local beliefs remain consistent with the converged global memory.} were correctly reflected in scoped memory without mismatch (Theorem~\ref{thm:prob_slice_bisim}). Following agent removal, no coherence violations or memory corruption were observed, validating Proposition~\ref{prop:local-failure}.

\begin{table}[ht]
\centering
\caption{Empirical Validation Results for a 250-Agent Search-and-Rescue Simulation}
\label{tab:exp-results}
\begin{tabular}{p{0.32\linewidth} p{0.52\linewidth} p{0.1\linewidth}}
\toprule
\textbf{Theorem / Proposition} & \textbf{Validation Metric} & \textbf{Score} \\
\midrule
Theorem~\ref{thm:semantic-coherence} (Semantic Coherence) & Invalid updates accepted & 0 / 11,325 \\
Theorem~\ref{thm:causal-isolation} (Causal Isolation) & Updates outside agent scope accepted & 0 \\
Theorem~\ref{thm:slice_bisim} (Stuttering Bisimulation) & Agent projection violations & 0 / 250 \\
Theorem~\ref{thm:prob_slice_bisim} (Probabilistic Bisimulation) & Validated proposals mismatched & 0 / 1,250 \\
Proposition~\ref{prop:local-failure} (Failure Containment) & Post-removal coherence violations & 0 \\
\bottomrule
\end{tabular}
\end{table}

Figure~\ref{fig:zoneA_timeline} illustrates a representative event sequence. Logs verified that throughout the run, agents actively executed their synchronized tasks: search agents continued to discover survivors, and rescue agents dispatched and completed recoveries, demonstrating that agents can exhibit coordinated behavior
solely through the semantics of validated updates and local policies.

\begin{figure}[ht]
  \centering
  \includegraphics[width=0.85\linewidth]{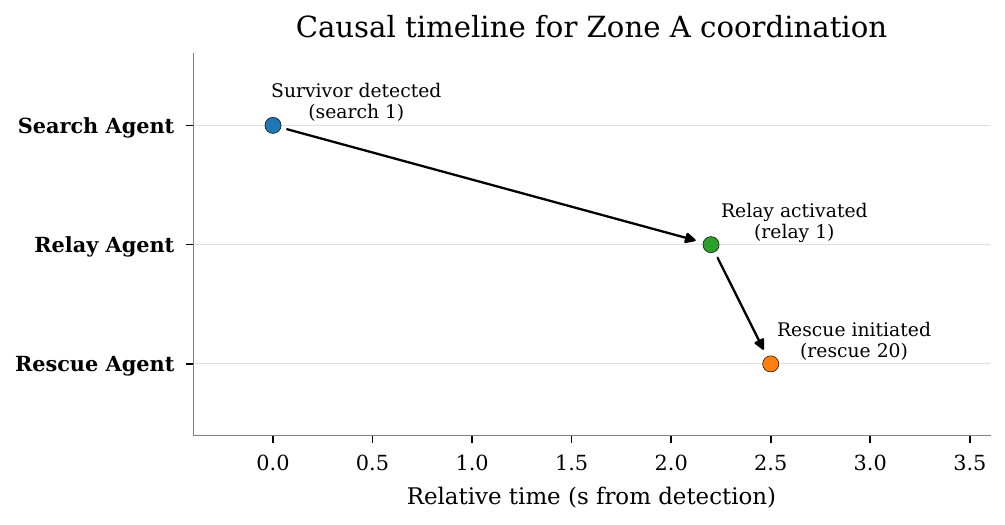}
  \caption{A survivor is detected by \texttt{search1}, triggering an ontology-valid update to shared memory. This change activates \texttt{relay1}, which moves to the zone and signals presence. Upon detecting relay support, \texttt{rescue20} initiates recovery. Each agent acts independently based on scoped semantic memory, without direct coordination or explicit messages.
}
  \label{fig:zoneA_timeline}
\end{figure}

\subsection{Communication Overhead Scaling}
\label{subsec:comm-scaling}

To empirically validate the message complexity bound from
Corollary~\ref{cor:bound-tightness}, we perform a controlled micro-benchmark
that directly models semantic propagation based on channel overlap. The theorem
establishes that communication cost for a validated update grows as
$\Theta(d)$—where $d$ is the number of agents whose semantic views intersect
the update—and not with the total agent population $n$. The experiment is
intentionally synthetic: it mirrors the theoretical abstraction of
independent prefix subscriptions and scoped dissemination, isolating the
semantic-overlap effect without conflating it with policy dynamics or update
frequencies arising in the full simulator.

Each ontology key is treated as a distinct semantic channel, and each agent
subscribes to a randomly chosen fraction $f$ of these channels. For each
$f \in [0,1]$, we generate 100 randomized slice assignments over a 
200-agent pool (100 relay, 100 rescue). We then publish an update on each key
and count the number of subscribing agents that would receive it under scoped
delivery. Adding the proposer, the expected message count per update is
$k + 1$, where $k = \mathrm{round}(f \cdot 200)$.

Figure~\ref{fig:comm-scaling} shows the results. The observed message cost
scales linearly with $f$, closely matching the ideal curve
$f \cdot 200 + 1$ across all slice fractions. At full overlap ($f = 1$), the
system recovers broadcast behavior (201 messages), while at low overlap
($f = 0.1$), communication cost decreases by more than 85\%. Variance is
highest at low $f$ due to the combinatorial distribution of slice membership
and vanishes as $f$ increases.

In real deployments, message loss or delay may occur. Because propagation is
slice-scoped rather than global, message redundancy remains inherently
bounded: only agents that subscribe to a given prefix must receive the
corresponding notification. In probabilistic channels, the expected message
volume increases by a factor $1/(1 - p)$, where $p$ is the drop probability,
preserving $\Theta(d)$ scaling for $d$-slice overlap. This ensures predictable,
graceful degradation under lossy conditions.

This benchmark therefore confirms that scoped propagation provides tight,
predictable control over communication cost even as the system scales.

\begin{figure}[ht]
  \centering
  \includegraphics[width=.7\linewidth]{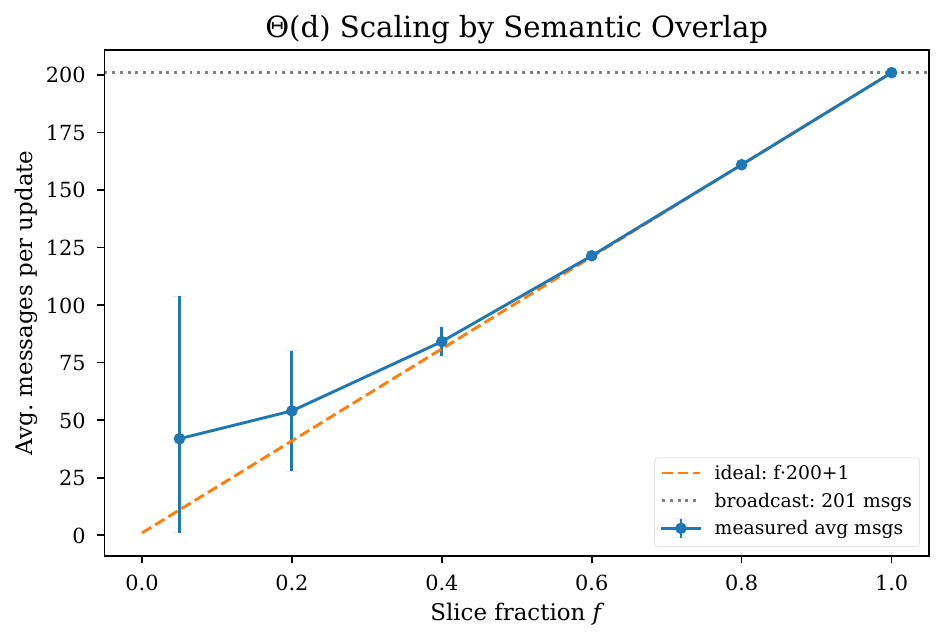}
  \caption{%
Message complexity scaling under semantic overlap. Each point shows the average number of agents receiving an update on a semantic key, as a function of the slice fraction $f$, averaged over 100 randomized slice assignments. Agents subscribe to a random $f$-fraction of ontology keys, and each update is received by all agents whose slice overlaps with the updated key. The expected message cost is $k+1$ with $k = \mathrm{round}(f \cdot 200)$. Observed values closely match the ideal bound $f \cdot 200 + 1$ (dashed line), with highest variance at low $f$ due to small slice sizes. The dotted line shows the broadcast baseline (201 messages).%
}
\label{fig:comm-scaling}
\end{figure}

\subsection{Geometric-Tail Convergence Validation}

To evaluate the exponential-delay bound predicted by Theorem~\ref{thm:dynamic-convergence}, we measured per-key alignment delays in a 100-step simulation under varying message delivery probabilities. Each agent maintained a locally scoped memory view and updated asynchronously based on validated proposals received through scoped refresh.

For each global memory update at time step $t$, we tracked the first future step $t' > t$ at which each agent's local view for the corresponding key matched the global memory. If such a match occurred, the alignment delay was recorded as $\Delta_b^{(t)} = t' - t$. This process captures the time it takes for an agent to realign with the global memory state, regardless of whether it received the original update or converged later through redundant proposals or indirect propagation.

Figure~\ref{fig:alignment_tail_plot} shows the empirical survival function $S(k) = \Pr[\Delta_b^{(t)} > k]$ across three delivery probabilities ($\rho = 0.2$, $0.5$, $0.8$). Each empirical curve is overlaid with a fitted exponential decay $e^{-\lambda k}$, where $\lambda$ is estimated independently for each trace. All curves exhibit tails consistent with exponential upper bounds, confirming the geometric decay predicted by Theorem~\ref{thm:dynamic-convergence}.

Notably, the estimated decay rates remain consistent across delivery probabilities, with $\lambda \approx 0.04$–$0.05$ for all runs. This stability suggests that convergence is driven not only by direct message reception, but also by structural redundancy and repeated proposal, indicating a form of self-healing propagation beyond simple i.i.d.\ delivery. The result empirically confirms the exponential bound and demonstrates that convergence occurs reliably and within bounded time under degraded yet recoverable conditions.

\begin{figure}[t]
  \centering
  \includegraphics[width=0.75\linewidth]{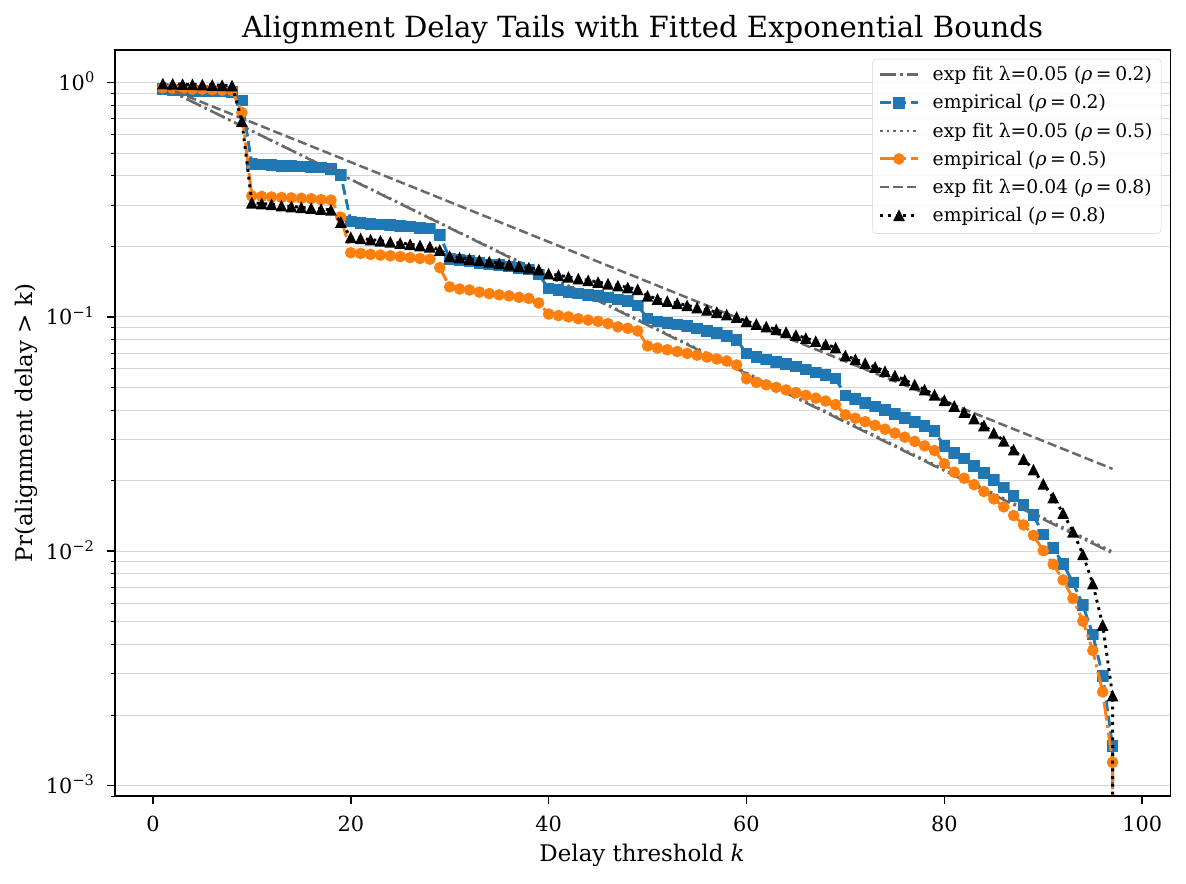}
  \caption{Empirical survival functions $S(k)$ vs.\ fitted exponential decays $e^{-\lambda k}$ for per-key alignment delays at different communication probabilities ($\rho$). All curves exhibit exponential tails, confirming the bound predicted by Theorem~\ref{thm:dynamic-convergence}. Minor mid-range deviation above the bound reflects conservative tail modeling and non-i.i.d.\ system structure, but asymptotic decay remains faster than the theoretical threshold.}
  \label{fig:alignment_tail_plot}
\end{figure}

\section{Discussion and Future Work}
\label{sec:discussion}

Semantic Fusion provides the shared semantic state that allows decentralized agents to operate with a consistent decision context, supporting coordinated behavior without centralized control. While these guarantees hold in the formal model, real deployments introduce factors not fully captured here, including asynchronous communication, partial failure, adversarial effects, and evolving ontologies. The discussion below considers SF’s core mechanisms—slice-local reasoning, scoped memory access, and decentralized validation—within these practical settings.

\subsection{Practical Deployment Considerations}
\label{subsec:comms-adversarial}

\paragraph{Robustness to Asynchrony and Adversarial Conditions.}
Semantic alignment supports coordinated behavior only when the operating picture remains sufficiently synchronized under imperfect communication. The formal model assumes reliable refresh propagation, but real systems face message loss, delays, and interference. Standard dissemination strategies, such as acknowledgment-based or epidemic protocols~\cite{demers1988epidemic}, can mitigate but not eliminate these issues. Convergence therefore becomes probabilistic, and a complete formal treatment of this regime remains open.

Trust assumptions may also vary. Agents may apply stronger validation to critical updates while accepting best-effort delivery for less consequential information. Cryptographic signatures~\cite{castro2002practical} or semantic endorsements can provide additional assurance in adversarial settings without requiring uniform trust across agents.

The model tolerates partial message loss, agent dropout, and asynchronous refresh without violating its semantic guarantees, indicating robustness to moderate degradation. As decentralized systems increasingly operate in contested or resource-limited environments, slice-based semantics offer a structured way to preserve coherence when communication is unreliable. One direction for future work is a probabilistic analogue of bisimulation in which local traces converge asymptotically, with high probability, to their projected global behaviors. We refer to this as \emph{probabilistic asymptotic bisimulation}.

The guarantees established here concern semantic safety: ensuring that memory evolution respects ontology constraints even under nondeterministic behavior or degraded communication. They do not address behavioral robustness or performance. An agent may act inefficiently while remaining semantically aligned with the global state. SF therefore separates semantic correctness from behavioral quality. In practice, SF provides the semantic foundation on which more resilient behaviors can be built, even as conditions degrade. Task-level efficiency or strategy selection, however, must be supplied by mechanisms beyond the core trace semantics.

\paragraph{Challenges of Scoped Delivery.}
The model assumes that refresh notifications can be routed only to agents whose slices intersect a modified entity. While realizable with content-based publish/subscribe systems, deploying strict scoped delivery at scale introduces engineering considerations outside the scope of the formal semantics. Maintaining per-entity routing state, updating slice membership as ontologies evolve, and ensuring correct delivery under high update turnover or partial failures all require additional infrastructure. Crucially, SF does not rely on a physical global store. The “global memory” is a semantic object reconstructed from the distributed ensemble of slice-local traces. Thus, correctness depends on slice semantics, not on a centralized memory substrate.

These concerns belong to the communication substrate rather than the formal model, but they influence how reliably agents maintain a coherent operating picture in practice. Evaluating SF under heterogeneous or degraded agent policies, where semantic consistency holds but behavioral quality varies, remains a useful direction for empirical study.

\subsection{Expressivity, Generality, and Ontology Evolution}
\label{subsec:onto-evolution}

\paragraph{Illustrative Application.}
Although our evaluation focuses on multi-agent coordination, the slice semantics are domain-agnostic. In distributed cyber defense, for example, agents can use slice-scoped state (such as \texttt{ThreatLevel@Host} or \texttt{ConnectionStatus}) to enforce local safety rules—for instance, isolating a host only when a connected peer has escalated threat status. By Theorem~\ref{thm:slice_bisim}, verifying such rules slice-locally suffices for global correctness, converting local checks into system-wide safety without centralized monitoring.

\paragraph{Ontology Evolution and Versioning.}
The current model assumes a fixed ontology, but long-lived deployments will require schema evolution. Supporting versioning, where slices are tagged and evolve incrementally, introduces challenges in consistency and migration. Future work may explore version snapshots, change propagation, and scoped migration protocols that maintain SF’s guarantees while allowing localized updates. Preserving coherent operating pictures under evolving ontologies is essential for practical use.

\paragraph{Generalization and Portability.}
The simulation instantiates one ontology and fixed roles, but neither is fundamental to SF. Slices can be redefined or extended without altering the semantics of validation or refresh propagation. Roles may be fixed or assigned dynamically. Update generators may be scripted, learned, or produced through generative models, provided their outputs satisfy ontology constraints. This flexibility allows SF to be applied across different domains while retaining its formal guarantees. Adapting SF to new settings requires only redefining the ontology, not modifying the core semantics.

\subsection{Behavioral Implications and Coordination Guarantees}
\label{subsec:expanded-reasoning}

\paragraph{Enabling Richer Local Reasoning.}
While Semantic Fusion enforces semantic coherence across distributed agents, it also permits expressive forms of local reasoning. These capabilities rely on SF to ensure that the semantic context within which decisions are made remains aligned across agents, grounding heterogeneous policies in a shared operating picture.

As agents incorporate learned or generative components whose outputs may vary across invocations, this distinction becomes important. SF imposes no assumptions on how such outputs are produced; it constrains only the admissibility of updates that enter shared memory. One direction for extension is \emph{policy-constrained behavior}, in which agents are assigned local contracts: guarded rules over their private state and slice-scoped memory that restrict the actions they may take. These contracts can be enforced independently by each agent while still ensuring that global properties are satisfied, provided the required convergence and coverage conditions hold.

This allows each agent to interpret and enforce its policy using only the information available in its slice. When combined with SF’s convergence guarantees, such contracts yield predictable behavior even when communication is delayed or partially unavailable. This enables formal guarantees on policy satisfaction, auditability, and compositionality, supporting scalable coordination in mission-critical settings.

Semantic Fusion therefore provides the underlying memory and trace semantics needed for higher-level mechanisms such as policy enforcement, delegated contracts, and forms of tacit coordination. Although these constructs do not appear in the core model, SF supplies the semantic substrate required to incorporate them safely, and they are developed further in follow-on work.

\paragraph{Slice Overlap and Convergence Tradeoff.}
Finally, Semantic Fusion exposes a natural tradeoff between communication efficiency and synchronization robustness. Smaller ontology slices reduce refresh overhead, yielding $\Theta(d)$ communication per update, but decrease the probability $\eta$ that a new update touches a given agent’s slice, slowing re-synchronization after missed updates. Broader slices increase overlap and speed convergence at the cost of greater communication. This flexibility allows slice design to balance efficiency and robustness depending on the application.

The core model guarantees slice convergence over update rounds, but prolonged quiescence can still delay re-synchronization. In practice, periodic heartbeats or light refresh probes can maintain alignment during inactive periods without requiring global coordination. These mechanisms preserve SF’s decentralized operation while improving practical convergence. Quantitative analysis of slice-overlap policies and their impact on dynamic convergence remains a promising direction for future work, particularly when update rates, observation patterns, or environmental conditions vary over time.

\section{Conclusion}\label{sec:conclusion}

This paper introduced Semantic Fusion, a formal framework for slice-local reasoning in decentralized multi-agent systems. By grounding agent updates in ontology-scoped memory slices and enforcing validation through shared semantics, we showed that agents can maintain global coherence and modularity without explicit message passing or centralized control. The architecture accommodates both symbolic and learned components while ensuring that updates remain structured, auditable, and semantically well-formed.

Our primary contribution is theoretical. We establish that agents operating over slice-scoped memory projections preserve semantic coherence (Theorem~\ref{thm:semantic-coherence}), converge under refresh (Theorem~\ref{thm:slice-consistency}), and remain causally isolated from irrelevant state changes (Theorem~\ref{thm:causal-isolation}). Most importantly, we show that local execution traces are stuttering-bisimilar to their global projections (Theorem~\ref{thm:slice_bisim}), enabling temporal properties to be verified slice-by-slice. We extend these results to probabilistic proposals (Theorem~\ref{thm:prob_slice_bisim}) and demonstrate that global safety properties expressible over ontology slices can be validated compositionally from local traces (Theorem~\ref{thm:slice-completeness-safety}). We further establish almost-sure slice convergence and corroborate this behavior empirically through dynamic re-synchronization (Theorem~\ref{thm:dynamic-convergence}). This establishes that structured, slice-scoped semantics can serve as a reliable foundation for decentralized reasoning, even when agents operate asynchronously or under partial observability.

To illustrate that the formal constructs introduced in Section~\ref{sec:core-constructs} can be realized in practice, we implemented a lightweight reference architecture that instantiates the SF model. While not a primary contribution, this implementation operationalizes the theoretical components and supports empirical validation of the core theorems. By instantiating ontology-scoped slices with structured validation and scoped refresh propagation, it provides a concrete demonstration of feasibility.

Semantic Fusion is designed to support heterogeneous and distributed agent systems, including those that incorporate learned or generative components. The framework provides a principled basis for decentralized execution, slice-scoped validation, and bounded communication under asynchronous or probabilistic conditions. By allowing agents to maintain alignment through shared semantics rather than centralized control, SF offers a robust foundation for verifiable reasoning in distributed settings.

\appendix

\section{Full Proofs}\label{app:proofs}

\subsection{Proof of Theorem~\ref{thm:slice_bisim} (Slice--Global Stuttering Bisimulation)}
\label{app:proof-thm10}

Let $\mathcal{G} = (Q,\;\xrightarrow{}\!)$ be the global LTS over memory states
$\mathcal{M}(t)$, and let $\mathcal{G}\!\mid_{O_a}$ be its projection onto
entities in the slice ontology~$O_a$.
Let $\mathcal{T}_a = (Q_a,\;\Rightarrow\!)$ be the local LTS of agent~$a$
over its semantic slice $S_a(t)=(O_a,M_a(t))$,
with transitions \textsc{prop} (validated local update) and
\textsc{sync} (refresh integration).

We assume:
\begin{enumerate}
  \item[\textbf{(S1)}] \emph{Slice--scoped validation.} Every integrated update is ontology--valid and affects only the authoring agent's slice.
  \item[\textbf{(S2)}] \emph{Reliable refresh propagation.} Each slice--relevant update eventually reaches every interested agent.
  \item[\textbf{(S3)}] \emph{Deterministic merge, no reordering.} Retrieved updates are merged exactly once in commit order.
  \textit{Remark.} This assumption can be satisfied in practice by enforcing a CRDT-style merge policy or applying timestamp-based total ordering per agent slice. Such mechanisms preserve the model's semantics while allowing decentralized update integration under asynchronous conditions.
\end{enumerate}

Then a stuttering bisimulation $\mathcal{R}_a \subseteq Q_a \times Q$ exists such that
\[
  \exists\,t' \le t:\;\bigl(M_a(t),\,\pi_{O_a}(\mathcal{M}(t'))\bigr)\in\mathcal{R}_a,
\]
so $\mathcal{T}_a \;\dot\sim_{\!\text{st}}\; \mathcal{G}\!\mid_{O_a}$.

\paragraph{Relation definition.}
For $t'\le t$, relate $q_a=M_a(t)$ to $q_G=\pi_{O_a}(\mathcal{M}(t'))$.  $\mathcal{R}_a$ is non--empty because the two stores are initialized identically at~$t=0$.

\paragraph{Label agreement.}
Since $M_a(t)\subseteq\pi_{O_a}(\mathcal{M}(t'))$, the valuations over $O_a$ coincide; hence $\mathcal{T}_a(q_a)=\mathcal{G}\!\mid_{O_a}(q_G)$.

\paragraph{Forth direction.}
\begin{itemize}
  \item \textsc{prop}: the same update appears immediately in $\mathcal{G}$; one global step satisfies the forth obligation.
  \item \textsc{sync}: the remote update was already committed globally at some $t^*\le t$, so a finite stutter in $\mathcal{G}$ leads from $q_G^*=\pi_{O_a}(\mathcal{M}(t^*))$ to $q_G$, satisfying the obligation.
\end{itemize}

\paragraph{Back direction.}
Identical reasoning: slice--relevant global steps eventually appear locally (\textbf{A2}); irrelevant steps correspond to local stutters.

\paragraph{Conclusion.} $\mathcal{R}_a$ meets both stuttering clauses, establishing the bisimulation. \qed

\subsection{Proof of Theorem~\ref{thm:prob_slice_bisim} (Probabilistic Slice--Global Bisimulation)}
\label{app:proof-thm12}

Let $\mathcal{P}_G = (Q_p,\xrightarrow{\phantom{p}})$ be the probabilistic LTS over global memory; $\mathcal{P}_G\!\mid_{O_a}$ is its projection.
Let $\mathcal{P}_a=(Q_a^p,\Rightarrow, \mathcal{L}_p)$ be the agent's probabilistic LTS with \textsc{prop} (sampling from $\gamma_a$) and \textsc{sync} (refresh).
Under slice--scoped validation, reliable refresh, and independent proposal sampling, there exists a probabilistic bisimulation $\mathcal{R}_a^p \subseteq Q_a^p \times Q_p$ such that
\[
  \exists\,t'\le t:\;\bigl(M_a(t),\,\pi_{O_a}(\mathcal{M}(t'))\bigr)\in\mathcal{R}_a^p,
  \quad\text{and}\quad
  \mathcal{P}_a \sim_{\text{pbis}} \mathcal{P}_G\!\restriction_{O_a}.
\]

\paragraph{1. Relation and coupling.}
For states $q_a=M_a(t)$ and $q_G=\pi_{O_a}(\mathcal{M}(t'))$ with $t'\le t$, define a coupling $\kappa$ between successor distributions $\mu_a$ and $\mu_G$:
\begin{itemize}
  \item \textsc{sync}$\,(\Delta S_b^{\text{valid}})$: pair with the \emph{same} global state $q_G'=q_G$ (probability~1) because the update is already present.
  \item \textsc{prop}: sample $u\sim\gamma_a(S_a)$.
    \begin{itemize}
      \item If $u$ is invalid, pair self--loops.
      \item If $u$ is valid, pair $q_a' = M_a(t)\cup\{u\}$ with $q_G' = \pi_{O_a}(\mathcal{M}(t+1))$ that integrates $u$, using probability $p_u$.
    \end{itemize}
\end{itemize}
$\kappa$ preserves total probability and has support in $\mathcal{R}_a^p$.

\paragraph{2. Label agreement.}
Subset relation ensures equal valuations over $O_a$.

\paragraph{3. Forth/back clauses.}
Coupling and Lemma~\ref{lem:delivery} establish matching probabilistic steps in both directions; irrelevant global transitions map to local stutters of measure~1.

\paragraph{4. Conclusion.}
$\mathcal{R}_a^p$ with couplings $\kappa$ satisfies the probabilistic bisimulation definition \cite{segala1995probbisim}. Hence $\mathcal{P}_a \sim_{\text{pbis}} \mathcal{P}_G\!\mid_{O_a}$. \qed

\subsection{Proof of Theorem~\ref{thm:slice-completeness-safety} (Slice--Local Completeness for Safety Properties)}
\label{app:proof-slice-completeness}

\paragraph{Restatement.}
For any safety formula $\varphi$ over atoms in $O_a$,
if all agent traces satisfy $\varphi$
($\forall a\in A,\, \mathcal{T}_a\models\varphi$),
then the global execution $\mathcal{G}$ satisfies $\varphi$.

\paragraph{Key lemmas.}
\begin{enumerate}
\item[\textbf{L1}] (\textit{Slice cover}.)
The family of ontologies $\{O_a\}_{a\in A}$ forms a cover of~$\mathcal{O}$:
every entity $e\in \mathcal{O}$ belongs to at least one slice.
(This is a system design assumption.)
\item[\textbf{L2}] (\textit{Stuttering preservation}.)
Safety properties in \textbf{LTL} \emph{and in the CTL$^\ast$ fragment
that omits the \emph{next} operator $X$}
are preserved under stuttering equivalence, as established in CTL and CTL* fragments
\cite{Emerson1986ModelCheck, kupferman2000modelcheck}.

\item[\textbf{L3}] (\textit{Projection bisimulation}.)
From Theorem~10,
$\mathcal{T}_a \simdot_{\mathrm{st}} \mathcal{G}\!\mid_{O_a}$.
\end{enumerate}

\paragraph{Proof.}
Assume, for contradiction,
that $\mathcal{G}\not\models \varphi$.
Then there exists a finite bad prefix
$\pi = G(0)\ldots G(k)$ that violates~$\varphi$
(safety $\Rightarrow$ finite counterexample).

By \textbf{L1},
every atomic proposition appearing in~$\pi$
is drawn from some ontology slice~$O_a$.
Let $a^\ast$ be an agent whose slice covers at least one
entity modified in $\pi$.
Projecting $\pi$ onto $O_{a^\ast}$
yields a finite trace $\pi'$
that still violates $\varphi$,
because $\varphi$ refers only to atoms in $O_{a^\ast}$.

By \textbf{L3} there exists a stuttering equivalent trace
in $\mathcal{T}_{a^\ast}$ matching $\pi'$.
By \textbf{L2} stuttering preserves safety counterexamples,
so $\mathcal{T}_{a^\ast}\not\models\varphi$,
contradicting the premise that all agents satisfy $\varphi$.

Therefore $\mathcal{G}\models\varphi$, completing the proof. \hfill\qedsymbol

\section{Reference Architecture Detail}\label{app:ref-arch}
\begin{algorithm}[H]
\caption{Asynchronous Agent Lifecycle in the SF Reference Architecture}
\begin{algorithmic}[1]
\State Initialize agent objectives
\State Construct semantic slice from ontology and initial memory
\While{task not complete}
    \State Reason asynchronously over local slice
    \If{structured update proposal is generated}
        \State Validate proposal against ontology
        \If{proposal is valid}
            \State Issue \emph{scoped} refresh notification to subscribers (entity \emph{IDs} only)
            \State Update local memory with validated change
        \EndIf
    \EndIf
    \State Listen for incoming refresh notifications
    \If{received notification intersects slice}
        \State Pull full update $\rightarrow$ revalidate $\rightarrow$ integrate
    \EndIf
\EndWhile
\end{algorithmic}
\end{algorithm}

\section{Experimental Methodology}

\subsection{Environment and Initial Conditions}

The simulation environment is a discrete rectangular grid of
$\texttt{world\_width} \times \texttt{world\_height}$ symbolic zones. Each zone
has a unique coordinate label and no hidden physical dynamics: all state
changes arise from agent-generated semantic updates, not from an exogenous
world model. The system contains no ground-truth survivor map; all survivor
state is induced entirely by search-agent proposals.

At initialization, zone coordinates are written to global memory and to the
local memories of agents whose ontology slices permit the corresponding
prefix. Search zones are assigned deterministically by shuffling the grid with
a fixed RNG seed and partitioning the zones evenly among the search agents.
Each search agent begins at the first zone in its assigned list.

\subsection{Ontology and Slice Configuration}

The ontology defines a small set of task-level semantic predicates used by
agents: \texttt{Survivor}, \texttt{ZoneStatus}, \texttt{Relay},
\texttt{Rescue}, \texttt{Bid}, \texttt{AgentPos}, and \texttt{ZoneCoord}.  Each
agent is given a slice that specifies which of these predicates it may read
and propose.

In the full-overlap configuration used for correctness tests, all search
agents observe \texttt{Survivor}, \texttt{ZoneStatus}, \texttt{AgentPos}, and
\texttt{ZoneCoord}; relay and rescue agents additionally observe and propose
\texttt{Relay}, \texttt{Rescue}, and \texttt{Bid}. For reduced-overlap
experiments, prefixes are distributed randomly across relay and rescue agents
while ensuring that each predicate has at least one owning agent.

Refresh propagation follows the selective-delivery semantics of the formal
model: an update on predicate $P$ is delivered only to agents whose slices
contain~$P$, and delivery succeeds with probability \texttt{comm\_prob}.

\subsection{Agent Observations, Actions, and Policies}

All agents operate asynchronously in discrete global ticks. At each tick, an
agent executes one local step consisting of: (i) reading its slice-local memory
state, (ii) sampling any internal stochastic triggers, and (iii) performing
slice-permitted semantic updates or movement actions.

\paragraph{Search agents.}
Each search agent moves deterministically through its assigned zone list.
At each zone, it samples a survivor-indication variable (detected with
probability~0.3, none with~0.7). If the update is slice-admissible, the agent
proposes \texttt{Survivor@zone = detected|none} and marks the zone as searched.

\paragraph{Relay agents.}
Relay agents read their slice for active survivor indications and move toward
zones requiring relay coverage. Upon arrival, a relay agent signals presence by
writing \texttt{Relay@zone = active}. Zone selection breaks ties stochastically
to avoid deterministic contention.

\paragraph{Rescue agents.}
Rescue agents compute bids for survivor zones based on distance and a small
random perturbation. A rescue proceeds only after the agent’s local slice shows
an active relay at the target zone and that the agent’s bid is highest among
slice-visible bids. Service duration is sampled from a fixed integer range,
after which the agent writes a rescue completion record and resets the zone
status.

\subsection{Asynchrony, Timers, and Execution Flow}

All agents run as independent asynchronous processes. The global tick advances
only once all agents complete their current step. Agents may be removed
mid-execution to test failure containment. After the main sequence of ticks,
the system performs several final “flush” ticks without new proposals to ensure
that outstanding refreshes are incorporated before analysis.

\subsection{Global Memory, Validation, and Logging}

Upon validation, the update is applied to the agent’s local memory and recorded in the global canonical store used for post hoc analysis. Other agents learn of the update only through the model’s selective refresh mechanism, which delivers it to agents whose slices intersect its scope. Invalid proposals (e.g., adversarial injections used to test
validator behavior) are rejected locally and do not modify any state.

The simulation records all proposed updates, deliveries, slice-level refreshes,
and local/global memory snapshots. These logs serve as the data source for the
bisimulation checks, semantic-coherence tests, and probabilistic-convergence
analyses reported in Section~6.

\subsection{Tunable Experimental Parameters}

All experiments rely on a configurable set of parameters: delivery probability
(\texttt{comm\_prob}), slice-overlap fraction (\texttt{fan\_out}), random seed,
agent counts, grid dimensions, duration, and schedule of invalid-update
injections. These parameters allow the scenario to be run in high-overlap or
sparse-overlap conditions, deterministically or with controlled stochasticity.

\section{Online Resources}
The code used to validate the theoretical results and conduct the simulation experiments is publicly available at: \url{https://github.com/fiazaich/SAR_MAS_sim}.

\bibliographystyle{ACM-Reference-Format}
\bibliography{sf-bibliography}

\end{document}